\documentclass[a4paper,UKenglish,cleveref, autoref,thm-restate]{lipics-v2021}
\pdfoutput=1

\title{Mean-Payoff-Parity and Lifting Strategies from MDPs to 2-Player Stochastic Games}
\titlerunning{Mean-Payoff Parity and Lifting Strategies}

\author{Mohan Dantam}{Department of Computer Science, University of Oxford, UK}{}{}{}
\author{Richard Mayr}{School of Informatics, University of Edinburgh, UK}{}{}{}

\authorrunning{M.~Dantam and R.~Mayr}
\Copyright{Mohan Dantam and Richard Mayr}

\relatedversiondetails{Full version of a paper presented at CONCUR 2026.}{}



\usepackage[colorinlistoftodos,prependcaption,textsize=tiny,obeyDraft]{todonotes}

\usepackage{amsmath,amssymb,amsthm}
\usepackage{mathtools}
\usepackage{enumitem}
\setlist[itemize]{parsep=2pt, topsep=2pt, partopsep=2pt}
\setlist[description]{parsep=2pt, topsep=2pt, partopsep=2pt}
\usepackage{xcolor,pgf,tikz,pgfplots}
\usepackage{comment}
\usepackage{hyperref}
\hypersetup{colorlinks,linkcolor=blue,citecolor=blue,urlcolor=blue}
\usepackage{cleveref}
\usepackage{array}
\usepackage{makecell}
\usepackage[ruled,vlined]{algorithm2e}
\usepackage{lmodern}
\usepackage{diagbox}
\usepackage{standalone}

\newcommand{\ie}{\emph{i.e.},\ }

\DeclareMathOperator{\attr}{\mathtt{Attr}}

\newcommand{\Eparity}{\mathtt{EPAR}}
\newcommand{\Oparity}{\mathtt{OPAR}}

\newcommand{\MP}{\operatorname{\mathtt{MP}}}
\newcommand{\mppar}{\MP > 0\, \cap \, \Eparity}
\newcommand{\energymove}[1]{\stackrel{#1}{\movesto}}

\newcommand{\tendsto}{\rightarrow}

\newcommand{\Q}{\mathbb{Q}}

\DeclarePairedDelimiter{\set}{\{}{\}}
\DeclarePairedDelimiter{\ceil}{\lceil}{\rceil}

\newcommand{\numevencolours}{\#\lrc{\mathrm{distinct\,even\,colors}}}

\newcommand{\+}[1]{\mathbb{#1}}
\newcommand{\?}[1]{\mathcal{#1}}

\newcommand{\N}{\+{N}}
\newcommand{\Z}{\+{Z}}
\newcommand{\R}{\+{R}}
\newcommand{\x}{\times}
\newcommand{\Ocompl}{\?{O}}
\newcommand{\Thetacompl}{{\Theta}}

\usepackage{wasysym} %
\usepackage{xparse}
\usepackage{xifthen}
\usepackage{mleftright}
\newcommand{\rsymbol}{\ocircle}
\newcommand{\zsymbol}{\Box}
\newcommand{\osymbol}{\Diamond}
\newcommand{\xsymbol}{\odot}
\newcommand{\zstates}{\states_\zsymbol}
\newcommand{\rstates}{\states_\rsymbol}
\newcommand{\ostates}{\states_\osymbol}
\newcommand{\xstates}{\states_\xsymbol}
\newcommand{\reachset}{T}

\newcommand{\abs}[1]{\lvert#1\rvert}
\newcommand{\card}[1]{\abs{#1}}

\newcommand{\size}[1]{\norm{#1}}

\newcommand{\norm}[2][]{\lVert#2\rVert_{#1}}

\newcommand{\eqby}[2][=]{\stackrel{\text{{\tiny{#2}}}}{#1}}
\newcommand{\eqdef}{\eqby{def}}

\usepackage{bm}

\newcommand{\eps}{\varepsilon}

\newcommand{\problemx}[3]{
\par\noindent\underline{\sc#1}\par\nobreak\vskip.2\baselineskip
\begingroup\clubpenalty10000\widowpenalty10000
\setbox0\hbox{\bf INPUT:\ }\setbox1\hbox{\bf QUESTION:\ }
\dimen0=\wd0\ifnum\wd1>\dimen0\dimen0=\wd1\fi
\vskip-\parskip\noindent
\hbox to\dimen0{\box0\hfil}\hangindent\dimen0\hangafter1\ignorespaces#2\par
\vskip-\parskip\noindent
\hbox to\dimen0{\box1\hfil}\hangindent\dimen0\hangafter1\ignorespaces#3\par
\endgroup}

\NewDocumentCommand{\Prob}{O{} O{} m}{\ifthenelse{\isempty{#3}}{\mathcal{P}^{#1}_{#2}}{\mathcal{P}^{#1}_{#2}\lrc{#3}}}
\NewDocumentCommand{\expectation}{O{} O{} m}{{\mathcal{E}^{#1}_{#2}\lrc{#3}}}
\NewDocumentCommand{\valueof}{O{} O{} m}{{\mathtt{val}^{#1}_{#2}\lrc{#3}}}
\NewDocumentCommand{\limval}{O{} O{} m}{{\mathtt{Lval}^{#1}_{#2}\lrc{#3}}}
\NewDocumentCommand{\ST}{O{} O{}}{\ifthenelse{\isempty{#2}}{\ifthenelse{\isempty{#1}}{\mathtt{ST}}{\mathtt{ST}\lrc{#1}}}{\mathtt{ST}\lrc{#1,#2}}}
\NewDocumentCommand{\winset}{O{} O{}}{\ifthenelse{\isempty{#2}}{\ifthenelse{\isempty{#1}}{\mathtt{Win}}{\mathtt{Win}_{#1}}}{\ifthenelse{\isempty{#1}}{\mathtt{Win}\lrc{#2}}{\mathtt{Win}_{#1}\lrc{#2}}}}
\NewDocumentCommand{\AS}{O{} O{} m}{{\mathtt{AS}^{#1}_{#2}\lrc{#3}}}

\newcommand{\dist}{\mathcal{D}}
\newcommand{\supp}{{\sf supp}}

\newcommand{\always}{{\sf G}}
\newcommand{\eventually}{{\sf F}}

\newcommand{\obj}{\mathtt{O}}

\newcommand{\hide}[1]{}

\newcommand{\lrb}[1]{\mleft[#1\mright]}
\newcommand{\lrc}[1]{\mleft(#1\mright)}
\newcommand{\lrd}[1]{\{#1\}}

\newcommand{\ignore}[1]{}

\newcommand{\nat}{\mathbb N}

\newcommand{\setcomp}[2]{\lrd{{#1} \mid {#2}}}
\newcommand{\given}{{\,\mid\,}}
\newcommand{\symdiff}{\Delta}
\newcommand{\tuple}[1]{\lrc{#1}}
\newcommand{\odd}{\text{Odd}}
\newcommand{\even}{\text{Even}}

\newcommand{\game}{{\mathcal G}}
\newcommand{\mdp}{{\mathcal M}}
\newcommand{\mc}{{\mathcal A}}

\newcommand{\gametuple}{\tuple{\states,(\zstates,\ostates,\rstates),\transition,\probp}}

\newcommand{\states}{S}

\newcommand{\s}{s}
\newcommand{\qstate}{q}

\newcommand{\transition}{{E}}

\newcommand{\movesto}{{\rightarrow}}

\newcommand{\probp}{P}

\newcommand{\complementof}[1]{\overline{#1}}

\newcommand{\play}{\rho}

\newcommand{\playsof}[1]{{\it Plays}\lrc{#1}}
\newcommand{\partialplay}{\rho}

\newcommand{\zstrat}{\sigma}
\newcommand{\ostrat}{\pi}
\newcommand{\xstrat}{\tau}
\newcommand{\optzstrat}{\zstrat^*}
\newcommand{\optostrat}{\ostrat^*}

\newcommand{\zallstrats}[1]{\zstratset^{{#1}}}
\newcommand{\oallstrats}[1]{\ostratset^{{#1}}}
\newcommand{\zfinstrats}[1]{\zstratset^{{#1}}_{\finite}}
\newcommand{\ofinstrats}[1]{\ostratset^{{#1}}_{\finite}}
\newcommand{\finite}{f}
\newcommand{\zstratset}{\Sigma}
\newcommand{\ostratset}{\Pi}
\newcommand{\px}{\xsymbol}
\newcommand{\pz}{\zsymbol}
\newcommand{\po}{\osymbol}

\newcommand{\memory}{{\sf M}}
\newcommand{\updatefun}{upd}
\newcommand{\memconf}{{\sf m}}
\newcommand{\memconfset}{{\sf M}}
\newcommand{\memsuc}{{\sf nxt}}

\newcommand{\memup}{{\sf \updatefun}}
\newcommand{\memstrattuple}{\tuple{\memory, \initmem, \memup,\memsuc}}
\newcommand{\initmem}{\memconf_0}

\newcommand{\om}{\omega}

\newcommand{\probm}{{\mathcal P}}

\newcommand{\coloring}{{\mathit{C}ol}}

\mathchardef\mhyphen="2D %

\newcommand{\F}{{\mathcal F}}

\DeclareMathOperator{\bits}{bits}

\newcommand{\successors}[1]{\mathsf{Succ}({#1})}

\usetikzlibrary{arrows,decorations,intersections,pgfplots.fillbetween}
\usetikzlibrary{arrows.meta,automata,shapes,decorations,positioning,trees,calc}
\usetikzlibrary{shapes.multipart,shapes.misc}
\usepgflibrary{decorations.pathreplacing}
\usetikzlibrary{arrows,topaths}
\usetikzlibrary{automata,shapes.arrows, backgrounds, fadings,intersections, positioning}
\usetikzlibrary{shadows}

\usetikzlibrary{decorations.pathmorphing}
\usetikzlibrary{decorations.shapes}
\usetikzlibrary{shapes}
\usetikzlibrary{backgrounds}
\usetikzlibrary{fit}
\usetikzlibrary{arrows}
\usetikzlibrary{calc}
\usetikzlibrary{automata}
\usetikzlibrary{shapes.geometric}

\pgfplotsset{compat=1.16}

\tikzset{every picture/.style={thick,>=angle 60, node distance=2cm and 2cm}}
\tikzset{Grand/.style={draw,circle,minimum size=11*1.5,inner sep=0}}
\tikzset{Gmax/.style={draw,rectangle,minimum size=9*1.5,inner sep=0}}
\tikzset{Gmin/.style={draw,diamond,minimum size=9*1.5,inner sep=0}}
\tikzset{gamebad/.style={fill=red}}

\tikzset{every loop/.style={looseness=20}}
\tikzset{>=stealth, shorten >=1pt, shorten <=1pt}

\tikzset{state/.style={draw,circle,inner sep=2,minimum size=17pt}}

\tikzset{square/.style={regular polygon,regular polygon sides=4}}
\tikzset{acta/.style={
fill=blue!40!white,square,inner sep=2,
}}
\tikzset{actb/.style={
fill=red!40!white,diamond,inner sep=1,
}}

\tikzset{elliptic state/.style={draw,ellipse, minimum width=30, minimum height=15}}

\nolinenumbers

\begin{document}

\maketitle

\begin{abstract}
  We consider the strategy complexity (i.e., memory and randomization)
  of optimal strategies in turn-based 2-player zero-sum stochastic games.
  Results in \cite{GimbertKelmendi:IJGT20203,MSTW2021} show how to lift
  optimal memoryless strategies for
  shift-invariant inverse-submixing objectives
  from MDPs to 2-player stochastic games with
  an exponential increase in the number
  of memory modes.
  We show the corresponding lower bound, i.e.,
  the extra exponential memory is required in general, even for randomized strategies.

  Moreover, we solve the strategy complexity of the well-studied
  mean-payoff-parity objective ($\mppar$) in 2-player stochastic games.
  This objective is also shift-invariant inverse-submixing,
  but easier than the worst case for this class.
  In MDPs, Maximizer has optimal \emph{memoryless randomized} strategies,
  while optimal \emph{deterministic} strategies require
  \emph{exponential} memory.
  However, in stochastic games, optimal \emph{randomized} strategies
  require, at least and at most, \emph{linear memory} (equal
  to the number of even colors).
  
  Finally, we show that the different construction in
  \cite{GZ2009,Bouyer:LMCS2023}
  for lifting memoryless (resp.\ finite-memory)
  \emph{deterministic} strategies from MDPs
  (resp.\ 1-player games) to 2-player games cannot be generalized even to
  memoryless \emph{randomized} strategies.
  We construct a shift-invariant objective where Max and Min each have
optimal memoryless randomized strategies in all MDPs,
but optimal (randomized) Max strategies still require \emph{infinite} memory
in deterministic 2-player games.
\end{abstract}


\section{Introduction}\label{sec:intro}

\subparagraph{Background on strategy complexity.}
We consider 2-player turn-based perfect information stochastic games
(here denoted as SGs) played on finite graphs.
They are also called
\emph{$2\frac{1}{2}$-player games} \cite{CJH2004,CJH2003} or
\emph{competitive Markov decision processes} \cite{Filar_Vrieze:book}. 
Introduced by Shapley \cite{shapley1953stochastic} in 1953, they
have since played a central role in the solution of many problems, e.g.,
synthesis of reactive systems
\cite{ramadge1987supervisory,pnueli1989synthesis}
and formal specification and verification
\cite{de2001interface,dill1989trace,alur2002alternating}.
Every state either belongs to one of the players (Max or Min)
or is a random state. In each round of the game the player who owns the current
state chooses the successor state along the game graph and
in random states the successor is chosen according to a predefined distribution.
Given a start state and strategies of Max and Min, this yields a
distribution over induced infinite plays.
Objectives are defined via a measurable reward function that assigns rewards
to plays (possibly just an indicator function of a measurable set),
and the players try to maximize (resp.\ minimize) the expected reward.

We consider the \emph{strategy complexity}, i.e., the required memory and randomization
of optimal strategies in SGs for a given objective.
In particular we study the question how the strategy complexity in SGs
compares to the strategy complexity in Markov decision processes (MDPs).

\subparagraph{Lifting strategies from MDPs to SGs for shift-invariant inverse-submixing objectives.}
\emph{Lifting} strategies from MDPs to SGs means adapting them to work even
in the more general context of games with a strategic opponent
\cite{GimbertKelmendi:IJGT20203,MSTW2021}.
Results about lifting typically address the strategy complexity, e.g.,
if optimal Max strategies for a certain objective require only memory of
size $x$ in MDPs then optimal Max strategies in SGs require only memory
of some size $y > x$, for a certain $y$, possibly depending on the objective and on the
size/structure of the SG.

For shift-invariant inverse-submixing objectives,
optimal and $\eps$-optimal finite-memory (randomized or deterministic)
Max strategies
can be lifted from MDPs to SGs with a \emph{doubly exponential increase}
in the number of memory modes (using $n$ extra bits of public memory in
binary-branching games where $n$ states are Min-controlled plus an additional
$2^n \cdot p$ bits where $p$ bits are required for optimal strategies in MDPs)
\cite[Theorem 1.2]{GimbertKelmendi:IJGT20203}.
(Recall that $x$ bits of memory means $2^x$ memory modes.)
A restricted subcase of this was observed in \cite[Theorem 6]{MSTW2021}.

{\bf Contribution 1:} When memoryless randomized strategies suffice for Max in
MDPs, the lifting above incurs an exponential increase in the number of memory
modes ($p=0$, $n$ extra bits). In \Cref{sec:lifting1}, we show the corresponding exponential lower bound
for the step from MDPs to games (even deterministic games) in this case.
This partly solves the open question in \cite[end of Sec.~6.1]{GimbertKelmendi:IJGT20203}.
We construct a class of examples for the multi-dimensional
$\MP > \vec{0}$ objective (which is shift-invariant inverse-submixing).
Though Max has optimal memoryless randomized strategies in all MDPs,
optimal Max strategies in deterministic 2-player games
require an exponential number of memory modes,
even if the strategies may use randomization and
even if Max's memory is private (hidden from the Min player).

\subparagraph{Strategy complexity of mean-payoff-parity.}
The mean-payoff-parity objective ($\mppar$) is to attain a strictly
positive mean payoff while also satisfying a parity objective.
This combines quantitative performance with a qualitative correctness
condition (encoded by parity), and it has frequently been studied, e.g.,
\cite{daviaud2018pseudo,CD2011,CHJ2005,chatterjee2011games,CDGQ:2014}.
$\mppar$ is shift-invariant inverse-submixing,
and hence suitable for the lifting of Max strategies from MDPs to SGs
as in \cite[Theorem 1.2]{GimbertKelmendi:IJGT20203}.
With deterministic strategies, Max already requires an exponential
number of memory modes even in MDPs \cite[Fig.~1]{Gimbert2011ComputingOS},
and hence the extra memory used in the lifting construction
($n + 2^n \cdot p$ extra bits of public memory, where $p$ grows polynomially
in the size of the game)
yields a doubly exponential upper bound on the number of memory
modes required in SGs.

{\bf Contribution 2:} We show that the situation is different for randomized strategies.
In \Cref{sec:mp_parity} we show that Max has optimal \emph{memoryless randomized} strategies
in MDPs. However, the lifting construction of \cite[Theorem 1.2]{GimbertKelmendi:IJGT20203}
would still yield exponentially many memory modes for Max strategies in SGs.
We show that optimal Max strategies for $\mppar$ in SGs (and also in
deterministic games)
require, at least and at most, just \emph{linearly} many memory modes,
equal to the number of even colors.
I.e., $\mppar$ is easier than the exponential worst case for the
lifting construction demonstrated in our lower bound in \Cref{sec:lifting1}.
This is an interesting example where randomization in strategies
drastically reduces the amount of memory required, but
does not eliminate the need for memory entirely.

Note however that none of this applies to the \emph{different} mean-payoff-parity
objective with a \emph{non-strict} inequality ${\MP \ge 0\, \cap \, \Eparity}$.
Optimal strategies for this objective require infinite memory even in
deterministic 1-player games (and thus also in MDPs/SGs) \cite{CHJ2005}, 
due to a pathological case where infrequent negative rewards at increasing intervals
(e.g., $-1$ at times $2^n$ for all $n \in \N$ and rewards zero otherwise)
can still yield a mean-payoff of zero.

\subparagraph{Lifting deterministic strategies from MDPs to SGs.}
A different lifting construction with orthogonal preconditions
was described in \cite[Theorem 9]{GZ2009}.
Given an objective, if the players have optimal memoryless
\emph{deterministic} strategies in all 
maximizing (resp.\ minimizing) MDPs, then they also have
optimal memoryless deterministic strategies in all SGs.
Under mild assumptions, this can be generalized to
finite-memory deterministic strategies
\cite[Theorem 4.1]{Bouyer:LMCS2023} (in stochastic or non-stochastic arenas).

{\bf Contribution 3:} We show that
such results do \emph{not} generalize to \emph{randomized} strategies,
even under very strong assumptions.
In \Cref{sec:other-lifting} we present stronger counterexamples than
the ones given in \cite{Bouyer:LMCS2023,Vandenhove:PhD}.
E.g., we construct a shift-invariant objective where Max and Min each have
optimal memoryless randomized strategies in all MDPs,
but optimal (randomized) Max 
strategies still require \emph{infinite} memory
in deterministic 2-player games.


\section{Preliminaries}\label{sec:prelim}

Let $\Z$ (resp.\ $\Z_{+},\N, \Q$) denote the set of integers (resp.\ positive, non-negative integers, rationals).
The `size' of an integer $n$ or a rational $q = \frac{m}{n}$
(where $\gcd\lrc{m,n}=1$) is defined in a natural way assuming binary representation.
Let $\size{n} \eqdef \lceil\log_2\lrc{\abs{n}+1}\rceil + 1$ and $\size{q} \eqdef \size{m} + \size{n} + 1$.
A \textit{probability distribution} over a countable set $S$ is a function
\mbox{$f: S \to [0,1]$} with \mbox{$\sum_{s \in S} f(s) = 1$}.
Let $\supp(f) \eqdef$ \mbox{$\setcomp{s}{f(s)>0}$} denote the support of
$f$ and $\dist(S)$ is the set of all probability distributions over $S$.
If $S$ is finite and $f(S) \subseteq \Q$, we define the \emph{bit size} of $f$
to be $\bits(f) \eqdef \sum_{\s \in S}\lrc{\size{f(\s)}}$.
Likewise for
probabilistic functions with rational probabilities $p: S \to 
\dist(T)$ with finite $S$ and $T$: $\bits(p) \eqdef 
\sum_{\s \in S}\bits(p(\s))$.

\ignore{
  Given an alphabet $\Sigma$,
  let $\Sigma^{\om}$ and $\Sigma^{*}$ ($\Sigma^+$) denote the set of infinite
  and finite (non-empty) sequences over $\Sigma$, respectively.
  Elements of $\Sigma^{\om}$ or $\Sigma^*$ are called words.
}

\subparagraph{Games, MDPs and Markov chains.}
We consider finite-state 2-player turn-based
perfect-information stochastic games (here abbreviated as SGs)
$\game = \gametuple$
where the finite set of states $\states$ is partitioned into the
states $\zstates$ of
the player $\pz$ (\emph{Maximizer}, for short Max),
states $\ostates$ of player $\po$ (\emph{Minimizer}, for short Min),
and chance vertices (aka random states) $\rstates$.
If $\rstates=\emptyset$ then it is called a deterministic 2-player game.
Let $\transition \subseteq \states \x \states$ be the transition relation.
We write $\s \movesto \s'$ if $\tuple{\s,\s'} \in \transition$ and
assume that
$\successors{\s} \eqdef \{\s' \mid \s\transition{}\s'\} \neq \emptyset$
for every state $\s$.
The \emph{probability function}~$\probp$
assigns each random state $\s \in \rstates$ a distribution over
its successor states, i.e., $\probp(\s) \in \dist(\successors{\s})$.
We extend the domain of $\probp$ to
$\states^*\rstates$ by $\probp(\partialplay\s) \eqdef \probp(\s)$
for all $\partialplay\s \in \states^+\rstates$.
An \emph{MDP} is a game where one of the two players does not control any
states. An MDP is \emph{maximizing} (resp.\ \emph{minimizing})
iff $\ostates = \emptyset$ (resp.\ $\zstates = \emptyset$).
A \emph{Markov chain} is a game with only random states,
i.e., $\zstates = \ostates = \emptyset$.

\subparagraph{Strategies.}
A \textit{play} is an infinite sequence $\s_0\s_1 \ldots \in \states^{\omega}$
such that $\s_i \movesto \s_{i+1}$ for all $i \ge 0$.
A \textit{path} is a finite prefix of a play.
Let $\playsof{\game} \eqdef \set*{\play = \lrc{\s_i}_{i \in \N} \, |
\s_i \movesto \s_{i+1}}$
denote the set of all possible plays.
A \textit{strategy} of the player $\pz$ ($\po$) is a function
$\zstrat\! : \states^* \zstates \to \dist(\states)$
($\ostrat\! : \states^* \ostates \to \dist(\states)$)
that assigns to every path
$w\s \in \states^* \zstates$ ($\in \states^* \ostates$)
a probability distribution over the successors of $\s$.
If these distributions are always Dirac then the strategy is called
\emph{deterministic} (aka pure), otherwise it is called \emph{randomized}.
The set of all strategies of player $\pz$ and $\po$ in
$\game$ is denoted by $\zallstrats{\game}$ and $\oallstrats{\game}$,
respectively.
A play/path $\s_0\s_1 \ldots$ is compatible with a pair of strategies
$(\zstrat,\ostrat)$ if $\s_{i+1} \in \supp(\zstrat(\s_0 \ldots \s_i))$
whenever $\s_i \in \zstates$ and
$\s_{i+1} \in \supp(\ostrat(\s_0 \ldots \s_i))$ whenever $\s_i \in \ostates$.

Memory-based strategies can be defined via \emph{Mealy Machines}~\cite{Mealy55}
as follows.
A strategy for a player $\px \in \{\pz,\po\}$ is a tuple
$\xstrat = \memstrattuple$ where $\memconfset$ is the set of memory modes,
$\initmem$ is the initial memory mode, the next/successor function
$\memsuc\! : \memconfset \x \xstates
\to \dist(\states)$ chooses a (distribution over) successor states based on the
current memory mode and state, and the update function $\memup\! : \memconfset 
\x \transition \to \dist(\memconfset)$ updates the memory mode upon observing
a transition.
Let $\xstrat[\memconf]$ denote the strategy $\xstrat$ starting
in memory mode $\memconf$ instead of $\initmem$.
Finite-memory Max (resp.~Min) strategies use a \emph{finite} set $\memconfset$
of memory modes and are denoted by $\zfinstrats{\game} (\text{resp. }\ofinstrats{\game})$.
Randomized memory updates are strictly more expressive than deterministic memory
updates for finite-memory strategies \cite{DBLP:conf/concur/MainR22}. 
Finite-memory strategies with deterministic/randomized next function and
deterministic/randomized update function are called FDD, FRD, FDR, FRR,
respectively, i.e., the first D/R refers to randomization
in the next function and second D/R refers to randomization in the memory update
function \cite[Figure 1.1]{DBLP:conf/concur/MainR22}
(adapted to our setting since we do not consider randomization over the
initial memory mode).
Note that, under standard definitions of behavioral strategies in game
theory, any memory in memory-based strategies is \emph{private by default},
since internal memory modes of the strategy of one player are not part of
the history of the play, and thus not visible to the other player(s).
Only in the special case of \emph{deterministic} strategies
there is no difference between public memory
(observable by the other player) and private memory,
since the other player can then accurately infer the current memory mode from the
known deterministic strategy and the observed history.
However, it can make a big difference for randomized strategies, where fixing
a private-memory strategy of one player yields only a partially observable MDP;
see, e.g., the Big Match, where the difference is crucial \cite{HIN:MOR2022}.

Strategies with memory $|\memconfset|=1$ are called \emph{memoryless}.
Memoryless deterministic (resp.\ randomized) strategies are called MD
(resp.\ MR).

The memory-size of a finite-memory strategy $\xstrat = \memstrattuple$
is defined as $\size{\xstrat} \eqdef \card{\memconfset}$. 
The total bit size of $\xstrat$ is defined as the size
of the rational probabilities used in $\xstrat$, i.e.,
$\bits(\xstrat) \eqdef \bits(\memsuc) + \bits(\memup)$.

\subparagraph{Measure.} A game $\game$ with initial state $\s_0$ and strategies
$(\zstrat,\ostrat)$ yields a probability space
$(\s_0\states^{\om},\F_{\s_0}, \Prob[\game][\zstrat,\ostrat,\s_0]{})$
where $\F_{\s_0}$ is the $\sigma$-algebra generated by the cylinder sets
$\s_0\s_1\ldots\s_n\states^{\om}$ for $n \ge 0$.
The probability measure $\Prob[\game][\zstrat,\ostrat,\s_0]{}$ is
first defined on the cylinder sets.
For $\partialplay = \s_0\ldots \s_n$, let
$\Prob[\game][\zstrat,\ostrat,\s_0]{\partialplay} \eqdef 0$ if
$\partialplay$ is not compatible with $\zstrat,\ostrat$ and
otherwise
$\Prob[\game][\zstrat,\ostrat,\s_0]{\partialplay\states^{\om}} \eqdef
\prod_{i=0}^{n-1} \xstrat(\s_0\ldots\s_i)(\s_{i+1}) $ where $\xstrat$
is $\zstrat$ or $\ostrat$ or $\probp$ depending on
whether $\s_i \in \zstates$ or $\ostates$ or $\rstates$, respectively.
By Carath\'eodory's extension
theorem~\cite{billingsley2008probability}, this defines a unique probability
measure on the $\sigma$-algebra.

\subparagraph{Objectives and payoff functions.}
General objectives are defined by real-valued measurable functions
$v\!: \s_0\states^{\om} \to \mathbb{R}$, and
we write $\expectation[][]{\cdot}$ for the expectation w.r.t.~$\probm$ and $v$.
For event-based objectives, $v$ is just the indicator function of
a measurable set of plays $\obj \subseteq \states^{\om}$, i.e., 
we consider the probability that plays belong to $\obj$.

An objective $v$ is called \emph{shift-invariant}
iff for all finite paths $\partialplay$ and plays $\play' \in \states^{\omega}$,
we have $v(\partialplay \play') = v(\play')$.
It has also been called \emph{prefix-independent}~\cite[Section~4]{DBLP:conf/concur/GimbertZ05}
or \emph{uniform}~\cite[Definition~3]{DBLP:journals/tcs/ColcombetN06} in prior
literature, but we follow the notation of~\cite{GimbertKelmendi:IJGT20203}
and use shift-invariant to avoid any ambiguities.
Objective $v$ is called \emph{submixing} iff for all sequences of finite
non-empty words $u_0$, $w_0$, $u_1$, $w_1 \ldots$ we have
$v(u_0w_0u_1w_1 \ldots) \le
\max(v(u_0u_1\ldots),v(w_0w_1 \ldots))$,
and \emph{inverse-submixing} iff
$v(u_0w_0u_1w_1 \ldots) \ge
\min(v(u_0u_1\ldots),v(w_0w_1 \ldots))$ \cite{GimbertKelmendi:IJGT20203}.

We assume standard notions about LTL, reachability and parity objectives
\cite{CGP:book}; cf.~\Cref{app:prelim}. $\Eparity$ denotes max-even parity.

\textit{Reward based objectives.}
Let $r\!: E \to \set{-R,\dots,0,\dots,R}^d$ be a bounded function that
assigns (possibly multi-dimensional) finite rewards to transitions.
If $\s \movesto \s'$ and $r((\s,\s')) = c$, we write $\s \energymove{c}
\s'$. Let $\play = \s_0 \energymove{c_0} \s_1 \energymove{c_1}
\ldots$ be a play.
We say that $\play$ satisfies
\textit{mean-payoff} $\MP > \vec{c}$ for some constant $\vec{c} \in \R^d$
iff $\lrc{\liminf_{n \tendsto \infty}\frac{1}{n}\sum_{i=0}^{n-1} c_i } > c$
(component-wise).
The \textit{mean-payoff-parity} objective is defined as $\mppar$.
Parity and mean payoff objectives are shift-invariant,
but reachability is not. Parity is submixing and inverse-submixing.

For finite-state SGs
there exist optimal MD Max strategies for
the reachability $\eventually\,\reachset$ objective \cite{CONDON1992203},
$\Eparity$ objective \cite{Zielonka:1998},
and $\MP >0$ objective for dimension one \cite[Prop.~7]{Brazdil2010}.

The size of a game $\game$ with an objective $v$ can be thought of as the
number of bits required to describe the states, transitions and probabilities
used by $\game$ along with the description of $v$.

\subparagraph{Determinacy.}
Given an objective $v$ and a game $\game$, state $\s$ \emph{has value}
(w.r.t.\ $v$) iff
$$\sup_{\zstrat \in \zallstrats{\game}}\inf_{\ostrat \in
\oallstrats{\game}} \expectation[\game][\zstrat,\ostrat,\s]{v} =
\inf_{\ostrat \in \oallstrats{\game}}\sup_{\zstrat \in
\zallstrats{\game}} \expectation[\game][\zstrat,\ostrat,\s]{v}.$$
If $\s$ has value then $\valueof[\game][v]{\s}$ denotes the value of
$\s$ defined by the above equality. A game with an objective is called
\textit{weakly determined} if every state has value.
Stochastic games with Borel objectives are weakly
determined~\cite{Maitra-Sudderth:2003,M1998}.
Our objectives above are Borel, hence any boolean combination of them is also
weakly determined \cite{CHJ2005}.
For $\eps > 0$ and state $\s$, a strategy
\begin{enumerate}
  \item $\zstrat \in \zallstrats{\game}$ is $\eps$-optimal
    (maximizing) iff $\expectation[\game][\zstrat,\ostrat,\s]{v} \ge
    \valueof[\game][v]{\s} - \eps$ for all $\ostrat \in \oallstrats{\game}$.
  \item $\ostrat \in \oallstrats{\game}$ is $\eps$-optimal
    (minimizing) iff $\expectation[\game][\zstrat,\ostrat,\s]{v} \le
    \valueof[\game][v]{\s} + \eps$ for all $\zstrat \in \zallstrats{\game}$.
\end{enumerate}
A $0$-optimal strategy is called \emph{optimal}.
An MD strategy is called \emph{uniformly} $\eps$-optimal
(resp.\ uniformly optimal)
if it is so from every start state.
Given an event-based objective $\obj$,
an optimal strategy for player $\pz$ from state $\s$
is \emph{almost surely} winning if $\valueof[\game][\obj]{\s}=1$.
By $\AS[\game][\pz]{\obj}$ we denote the set of states that have an almost
surely winning strategy for $\pz$ for objective $\obj$.
We drop subscripts and superscripts wherever possible if they are
clear from the context.

\section{Lifting Strategies from MDPs to 2-Player Stochastic Games for 
Shift-Invariant Inverse-Submixing Objectives}\label{sec:lifting1}

The \emph{finite memory transfer theorem} of Gimbert \& Kelmendi
\cite[Theorem 1.2]{GimbertKelmendi:IJGT20203} shows that finite-memory
Min (resp.\ Max) strategies
can be lifted from MDPs to 2-player stochastic games
if the payoff function is both shift-invariant and submixing
(resp.\ inverse-submixing).
This is a consequence of the following slightly stronger theorem.
(Adapted to our notation, because we consider objectives
from Max's point of view.)

\begin{theorem}[{\cite[Theorem 6.1]{GimbertKelmendi:IJGT20203}}]\label{thm:fmt}
  Let $f$ be a shift-invariant and inverse-submixing payoff function.
  If, for all $\eps > 0$, Max has an $\eps$-optimal strategy
  with finite memory in every finite-state MDP,
  then in every finite-state turn-based
  two-player stochastic game he has an
  $\eps$-subgame-perfect strategy that has finite memory.
  
  The statement also holds for $\eps = 0$, that is: if Max has a
  finite-memory optimal strategy in every game controlled by himself,
  then in every two-player game
  he has an optimal subgame-perfect strategy with finite memory.
\end{theorem}

Subgame-perfect strategies are a stronger notion than optimal strategies and
it suffices to view them as optimal strategies for our results.
The proof of \cite[Theorem 6.1]{GimbertKelmendi:IJGT20203}
also shows that deterministic (resp.\ randomized) Max strategies in MDPs are
lifted to deterministic (resp.\ randomized) Max strategies in the game.
While \cite{GimbertKelmendi:IJGT20203} only consider strategies with
deterministic memory updates, their construction can also lift
(from MDPs to games) strategies with randomized memory updates.
However, in \cite[Theorem 6.1]{GimbertKelmendi:IJGT20203},
the \emph{extra} memory bits used by the constructed Max
strategy in the game are always updated deterministically.
Thus this Max strategy always uses
\emph{public} memory, 
since the memory in the MDP strategies was trivially public
(by the absence of an opponent).

The size of the memory used by Max's strategy in the 2-player game
can be described as follows. For each MD strategy choice for Min,
Max fixes this strategy and computes the optimal strategy in this derived MDP
and correspondingly reserves the required memory space to reliably play this
strategy (think of this as corresponding to remembering how to play when Min
eventually fixes to this strategy). In addition to this, Max also remembers
the choice Min played last at each of its states. Assuming $k$ is the number
of Min-controlled states, $d$ the maximal out-degree of these states and $p$
to be the maximum number of memory bits needed to play in any of the derived
MDPs, then the total number of bits used by the described strategy would be of
the order $k\cdot \ceil{\log_2 d} + d^k \cdot p$.
Thus, in general, Max uses a \emph{doubly exponential} number of memory modes
in the size of the game, even if his good strategies in all MDPs use just polynomially many memory modes;
cf.~\cite[Sec.~6.1]{GimbertKelmendi:IJGT20203}.

An interesting subcase is to look at the construction when $p=0$, i.e.,
when optimal strategies for Max in MDPs are memoryless. In this case, the above
construction implies a memory bound of the order $k \cdot \ceil{\log_2 d}$
which is only exponential w.r.t.\ the game size. The end of~\cite[Sec.~6.1]{GimbertKelmendi:IJGT20203}
claims that this is unnecessary, since a different lifting result~\cite{GZ2009}
implies that memoryless strategies suffice, even in games. We show that this is true
only for deterministic strategies and does not generalize to strategies that are
memoryless randomized.
In particular, we show the matching exponential lower bound in \Cref{thm:trigger-lower}.
Even if Max has optimal memoryless randomized strategies in all MDPs,
in general he still needs exponential memory in the 2-player game, even if the memory is private.
This attempts to partly tighten the lower bound at
\cite[end of Sec.~6.1]{GimbertKelmendi:IJGT20203} by showing that some parts
of the construction are necessary when considering randomized strategies.
On the other hand, for some particular objectives, the construction in
\cite[Theorem 6.1]{GimbertKelmendi:IJGT20203} uses more memory than
necessary. In \Cref{sec:mp_parity} we show that for the $\mppar$ objective
Max requires, at least and at most, just a \emph{linear} number of
memory modes in SGs (and the lower bound holds even for deterministic games).

\begin{figure}[t]
   \vspace*{-17mm} 
  \scalebox{0.5}{
  \begin{tikzpicture}[->,>=stealth',shorten >=1pt,auto,node
    distance=2.0cm, font=\Large,
    max_node/.style={Gmax, minimum size=12mm},
    min_node/.style={Gmin, minimum size=16mm},
    edge_label/.style={below, sloped},
    edge_label2/.style={below, sloped}]
    
    \node[min_node] (s1) {$s_1$};
    \node[min_node, right=of s1] (s2) {$s_2$};
    \node[right=of s2] (dots_s) {$\dots$};
    \node[min_node, right=of dots_s] (sk) {$s_{k}$};

    \node[max_node, right=of sk] (t1) {$t_1$};
    \node[max_node, right=of t1] (t2) {$t_2$};
    \node[right=of t2] (dots_t) {$\dots$};
    \node[max_node, right=of dots_t] (tk) {$t_{k}$};
    
    \node[min_node, above right=of s1] (s1L) {$s_{1,L}$};
    \node[min_node, below right=of s1] (s1R) {$s_{1,R}$};
    
    \node[min_node, above right=of s2] (s2L) {$s_{2,L}$};
    \node[min_node, below right=of s2] (s2R) {$s_{2,R}$};
    
    \node[min_node, above right=of sk] (skL) {$s_{k,L}$};
    \node[min_node, below right=of sk] (skR) {$s_{k,R}$};

    \node[max_node, above right=of t1] (t1L) {$t_{1,L}$};
    \node[max_node, below right=of t1] (t1R) {$t_{1,R}$};

    \node[max_node, above right=of t2] (t2L) {$t_{2,L}$};
    \node[max_node, below right=of t2] (t2R) {$t_{2,R}$};
    
    \node[max_node, above right=of tk] (tkL) {$t_{k,L}$};
    \node[max_node, below right=of tk] (tkR) {$t_{k,R}$};

    \draw[->] ++(-1.25,0) -- (s1); 

    \path[->, >=latex]
    (s1) edge node[midway, edge_label]{$\vec{e}_{1}$} (s1L)
    (s1) edge node[midway, edge_label]{$- \vec{e}_{1}$} (s1R)
    (s1L) edge node[midway, edge_label]{$\vec{0}$} (s2)
    (s1R) edge node[midway, edge_label]{$\vec{0}$} (s2)
    (s2) edge node[midway, edge_label]{$\vec{e}_{2}$} (s2L)
    (s2) edge node[midway, edge_label]{$- \vec{e}_{2}$} (s2R)
    (sk) edge node[midway, edge_label]{$\vec{e}_{k}$} (skL)
    (sk) edge node[midway, edge_label]{$-\vec{e}_{k}$} (skR)
    
    (s2L) edge[dotted] (dots_s)
    (s2R) edge[dotted] (dots_s)
    (dots_s) edge[dotted] (sk)
    
    (skL) edge node[midway, edge_label]{$\vec{0}$} (t1)
    (skR) edge node[midway, edge_label]{$\vec{0}$} (t1)
    
    (t1) edge node[midway, edge_label]{$- \vec{e}_{1}$} (t1L)
    (t1) edge node[midway, edge_label]{$\vec{e}_{1}$} (t1R)
    (t1L) edge node[midway, edge_label]{$\vec{0}$} (t2)
    (t1R) edge node[midway, edge_label]{$\vec{0}$} (t2)
    (t2) edge node[midway, edge_label]{$- \vec{e}_{2}$} (t2L)
    (t2) edge node[midway, edge_label]{$\vec{e}_{2}$} (t2R)
    (tk) edge node[midway, edge_label]{$-\vec{e}_{k}$} (tkL)
    (tk) edge node[midway, edge_label]{$\vec{e}_{k}$} (tkR)
    
    (t2L) edge[dotted] (dots_t)
    (t2R) edge[dotted] (dots_t)
    (dots_t) edge[dotted] (tk);

    \draw[->, >=latex] (tkL) to[out=165, in=70]  node[midway, edge_label]{$\vec{\delta}$} (s1);
    \draw[->, >=latex] (tkR) to[out=195, in=290] node[midway, edge_label]{$\vec{\delta}$} (s1);
\end{tikzpicture}
}
  \vspace*{-17mm}
  \caption{In the game $\game_k$, optimal Max strategies for $\MP > \vec{0}$
  for $2k$ dimensions require
  at least $2^k$ memory modes. We have $2k$-dimensional reward vectors 
  $\vec{e_i}$ such that $\vec{e_i}[2i-1]=+1$, $\vec{e_i}[2i]=-1$ and $0$
  elsewhere.
  The special vector $\vec{\delta}=(\delta,\dots,\delta)$ has value
  $\delta\eqdef 2^{-2k}$ in every dimension.
\vspace*{-2mm}}
\label{fig:trigger-lower}
\end{figure}
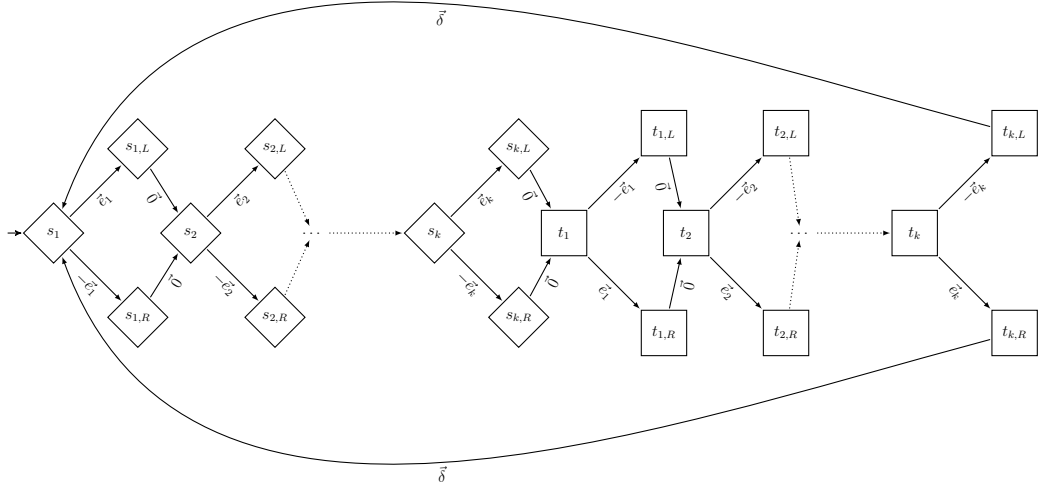

The multi-dimensional $\MP > \vec{0}$ objective
is shift-invariant (by definition) and inverse-submixing
(by \cite[Lemma~6]{BBE2010:arxiv} for dimension $1$,
and the fact that event-based inverse-submixing objectives are closed under
intersection).
Moreover, by~\cite[Prop.~5.1]{Kuceraetal2014}, 
almost surely winning strategies for
$\MP > \vec{0}$ in MDPs can be chosen as memoryless randomized.
However, the games in \Cref{fig:trigger-lower}
(similar to
\cite[Fig.~4]{DBLP:journals/acta/ChatterjeeRR14} but not identical)
show that even randomized Max strategies
need exponentially many memory modes.

\begin{restatable}{theorem}{triggerlower}\label{thm:trigger-lower}
There is a family of deterministic games $\game_k$ with $6k$ states
as in \Cref{fig:trigger-lower},
where every state is almost surely winning for the
$2k$-dimensional $\MP > \vec{0}$ objective 
for Max, but any
randomized Max strategy with $< 2^k$ private memory modes cannot win almost surely.
\end{restatable}
\newcommand{\triggerlowerproof}{
\begin{proof}
Consider the games $\game_k$ from \Cref{fig:trigger-lower}
with $6k$ states.
The dimension of the rewards $2k$ is split into $k$ blocks of $2$ dimensions each
and the reward vector $\vec{e_i}$ is $(+1,-1)$ on the $i$-th block
and zero elsewhere.
First Min makes $k$ decisions, choosing between $\vec{e_i}$
and $\vec{-e_i}$ for each $i=1,2,\dots,k$.
Then Max makes $k$ decisions, choosing between $\vec{-e_i}$
and $\vec{e_i}$ for each $i=1,2,\dots,k$.
Max can win $\MP > \vec{0}$ surely (from $\s_1$ and thus from every other state)
by exactly copying Min's choices such that these rewards cancel
out, which leaves just the strictly positive reward vector
$\vec{\delta} = (2^{-2k}, \dots, 2^{-2k}) > \vec{0}$ between each visit to $\s_1$.
(Since $\vec{\delta}$ can be described with a number of bits
that is polynomial in $k$, $\game_k$ has polynomial size.)

Intuitively, Max needs at least $2^k$ memory modes
to remember which of the $2^k$ possible choices Min
made, in order to copy it.
While this is easy to prove for deterministic Max strategies,
we show a stronger result:
Even randomized Max strategies with $< 2^k$ private memory modes
cannot win almost surely.

Let $\zstrat$ be an arbitrary randomized Max strategy with $m < 2^k$ private memory modes.
Overall, for the $k$ choices in the states $\s_1,\dots,\s_k$,
Min has $2^k$ different possible options,
which we denote by index numbers $j \in \{1,\dots,2^k\}$.
(Similarly for Max's options in states $t_1,\dots,t_k$.)
Let $\ostrat_j$ be the deterministic Min strategy that always plays option $j$ forever.
Let $\ostrat$ be the randomized Min strategy which initially picks
a $j \in \{1,\dots,2^k\}$, with equal probability $2^{-k}$, and then plays
$\ostrat_j$ forever.
(Note that $\ostrat$ does not care about Max's strategy at all, and thus
neither about Max's memory modes, so the memory can be private.)
Now we show that $\Prob[\game_k][\zstrat,\ostrat,\s_1]{\MP > \vec{0}} < 1$.

Let $M_k^j$ be the Markov chain derived from $\game_k$ by fixing
$\zstrat$ and $\ostrat_j$.
If $M_k^j$ is irreducible and aperiodic then
let $P$ be the $2^k \times m$ matrix such that
$P(j,i)$ is the probability, in the steady-state,
that $\zstrat$ is in memory mode $i \in \{1,\dots,m\}$
in state $t_1$.
(It is also possible that $M_k^j$ is reducible,
since it has a memory component from the strategy $\zstrat$
as part of its state.
If $\zstrat$ is such that it would never visit certain memory modes again eventually,
then this makes $M_k^j$ reducible.
In this case we consider an irreducible subchain
that is reached with maximal probability.
If $M_k^j$ is periodic then consider the long-run average frequency
of being in memory mode $i$ in state $t_1$ instead of the probability in the
following arguments.)
Let $Q$ be the $m \times 2^k$ matrix where $Q(i,j)$ is the
probability that $\zstrat$ will play option $j$ in states $t_1,\dots,t_k$
if it is in memory mode $i$ in state $t_1$.

The probability that Max will play option $j$ in $M_k^j$
in the long run
is then given by $(PQ)(j,j)$, the $j$-th diagonal value in the product matrix.
Since $m < 2^k$, the $2^k \times 2^k$ matrix $PQ$ does not have full rank.
Thus at least one of its eigenvalues is $0$.
Moreover, since both matrices $P$ and $Q$ are row-stochastic, the absolute
value of the eigenvalues 
of $PQ$ are upper-bounded by $1$. So the sum of the eigenvalues of $PQ$
is upper-bounded by $2^k -1$.
The Cayley–Hamilton theorem implies that the trace is the sum of the
eigenvalues. This yields an upper bound on the  
trace of $PQ$, namely ${\it Tr}(PQ)= \sum_{j=1}^{2^k} (PQ)(j,j) \le 2^k -1$.
Hence there exists a $j'$ such that $(PQ)(j',j') \le 1-2^{-k}$.
So, in $M_k^{j'}$, the long-run probability that Max will pick
some option different from $j'$ is at least $2^{-k}$.
Hence there exists at least one option $j'' \neq j'$ that
is picked by Max in the long run with probability $\ge 2^{-k}2^{-k} = 2^{-2k}$.
Since $j'' \neq j'$, there exists at least one block of two dimensions,
say $x$ and $x+1$, 
where Max choosing $j'$ has effect $(+1,-1)$ and $j''$ has the opposite effect
$(-1,+1)$ (or vice-versa; this case is symmetric). 
Thus, between two visits to state $s_1$, the expected reward in dimension
$x$ is $\le -2\cdot 2^{-2k} + \delta = -2^{-2k} < 0$.
Therefore, in $M_k^{j'}$, the $\MP > \vec{0}$ objective is satisfied with
probability zero, since it almost surely fails in dimension $x$.
I.e., $\Prob[M_k^{j'}][\zstrat,\ostrat_{j'},\s_1]{\MP > \vec{0}}=0$.

Since Min's strategy $\ostrat$ initially decides to become
$\ostrat_{j'}$ with probability $2^{-k}$,
we obtain that $\Prob[\game_k][\zstrat,\ostrat,\s_1]{\MP > \vec{0}}
\le 1- 2^{-k} (1-\Prob[M_k^{j'}][\zstrat,\ostrat_{j'},\s_1]{\MP > \vec{0}})
= 1 - 2^{-k} < 1$.
\end{proof}
}
\triggerlowerproof

As a side note, the assumption of a shift-invariant inverse-submixing Max objective
implies that Min always has optimal memoryless deterministic (MD) strategies in all MDPs/SGs by
\cite[Theorem 1.1]{GimbertKelmendi:IJGT20203}, since this objective is then
shift-invariant and submixing for Min.

\section{Strategy Complexity of \texorpdfstring{$\mppar$}{Positive Mean-payoff Parity} With Randomization}\label{sec:mp_parity}

In this section, unless otherwise stated, we consider the 
one-dimensional $\mppar$ objective, where Max needs
to attain a strictly positive one-dimensional mean payoff while
also satisfying a max-even parity condition.
These are defined via a reward function $r$ and a coloring function $\coloring$
where rewards are in $[-R,R]$ and colors in $\set*{0,1, \ldots, d}$.

For almost surely winning strategies in MDPs $\mdp$,
\emph{deterministic} Max strategies require
$\Ocompl\lrc{\exp\lrc{\size{\mdp}}}$ memory modes, i.e.,
exponential memory is both necessary and sufficient
\cite[Theorem~5]{Gimbert2011ComputingOS}.
We show that \emph{randomized} Max strategies require less memory: none in MDPs
and just \emph{linearly} many memory modes in stochastic games.

\begin{restatable}{theorem}{MRmpparMDP}\label{thm:MR_mppar_MDP}
In maximizing MDPs, almost surely
winning strategies for the multi-dimensional $\MP > \vec{0} \, \cap \, \Eparity$ objective
can be chosen memoryless randomized (MR).
\end{restatable}
\newcommand{\MRmpparMDPproof}{
\begin{proof}
We use the fact that, for any strategy, almost surely all the runs eventually end up in an
end component~\cite[Theorem~3.2]{D-A1997}. In a finite-state MDP, there can
only be a finite (albeit exponential) number of end components.
Let
$\mdp_1 \eqdef \tuple{\states_1, \zstates^1, \rstates^1, \transition_1, \probp_1}$ 
be one such end component of $\mdp$.
We say that $\mdp_1$ is `winning' iff the highest color is even
and $\states_1 \subseteq \AS[\mdp_1][\pz]{\MP > \vec{0}}$.
Given an almost surely winning strategy
$\optzstrat$, except for a null set, all induced runs must eventually
stay inside a winning end component.
Thus from every almost surely winning state it is possible to almost surely
reach a winning end component.
\begin{claim}\label{claim:multi-mppar-mdp}
    Let $\mdp_1$ be a winning end component. Then there is a memoryless
    randomized strategy $\optzstrat$ which is almost surely winning for
    $\MP > \vec{0} \, \cap \, \Eparity$ from every state in $\mdp_1$.
\end{claim}
\begin{claimproof}
    By definition of a winning end component, we know that 
    $\states_1 \subseteq \AS[\mdp_1][\pz]{\MP > \vec{0}}$. From the proof
    of~\cite[Prop.~5.1 (ii)]{Kuceraetal2014} we obtain that
    there exists a memoryless randomized strategy $\xi_\eps$ for some $\eps > 0$,
    such that $\xi_\eps$ almost surely wins $\MP > \vec{0}$ from any state in
    $\states_1$
    and $\size{\eps}$, the size of the probabilities used
    by $\xi_\eps$, is polynomial in $\size{\mdp_1}$.
    Furthermore, every edge in
    $\transition_1$ is used with some positive probability. This implies that
    $\xi_\eps$ also satisfies $\Eparity$ almost surely, since the highest color in $\mdp_1$ is even.
\end{claimproof}
Consider the union of all winning
end components $\?C$ in $\mdp$.
We obtain an almost surely winning MR strategy for $\MP > \vec{0} \, \cap \, \Eparity$
from every state in $\mdp$ as follows.
Outside of $\?C$ it plays an almost surely winning uniform MD strategy for
the reachability objective $\eventually\, \?C$.
Inside each maximal winning end component $\mdp_1$ in $\?C$ it plays
the MR strategy $\xi_{\eps}$ from \Cref{claim:multi-mppar-mdp}.

\end{proof}
}

\begin{proof}
We use the fact that, for any strategy, almost surely all the runs eventually end up in an
end component~\cite[Theorem~3.2]{D-A1997}. In a finite-state MDP, there can
only be a finite (albeit exponential) number of end components.
Let
$\mdp_1 \eqdef \tuple{\states_1, \zstates^1, \rstates^1, \transition_1, \probp_1}$ 
be one such end component of $\mdp$.
We say that $\mdp_1$ is `winning' iff the highest color is even
and $\states_1 \subseteq \AS[\mdp_1][\pz]{\MP > \vec{0}}$.
Given an almost surely winning strategy
$\optzstrat$, except for a null set, all induced runs must eventually
stay inside a winning end component.
Thus from every almost surely winning state it is possible to almost surely
reach a winning end component.
\begin{claim}\label{claim:multi-mppar-mdp}
    Let $\mdp_1$ be a winning end component. Then there is a memoryless
    randomized strategy $\optzstrat$ which is almost surely winning for
    $\MP > \vec{0} \, \cap \, \Eparity$ from every state in $\mdp_1$.
\end{claim}
\begin{claimproof}
    By definition of a winning end component, we know that 
    $\states_1 \subseteq \AS[\mdp_1][\pz]{\MP > \vec{0}}$. From the proof
    of~\cite[Prop.~5.1 (ii)]{Kuceraetal2014} we obtain that
    there exists a memoryless randomized strategy $\xi_\eps$ for some $\eps > 0$,
    such that $\xi_\eps$ almost surely wins $\MP > \vec{0}$ from any state in
    $\states_1$
    and $\size{\eps}$, the size of the probabilities used
    by $\xi_\eps$, is polynomial in $\size{\mdp_1}$.
    Furthermore, every edge in
    $\transition_1$ is used with some positive probability. This implies that
    $\xi_\eps$ also satisfies $\Eparity$ almost surely, since the highest color in $\mdp_1$ is even.
\end{claimproof}
Consider the union of all winning
end components $\?C$ in $\mdp$.
We obtain an almost surely winning MR strategy for $\MP > \vec{0} \, \cap \, \Eparity$
from every state in $\mdp$ as follows.
Outside of $\?C$ it plays an almost surely winning uniform MD strategy for
the reachability objective $\eventually\, \?C$.
Inside each maximal winning end component $\mdp_1$ in $\?C$ it plays
the MR strategy $\xi_{\eps}$ from \Cref{claim:multi-mppar-mdp}.

\end{proof}

Since both $\MP > 0$ and $\Eparity$ are shift-invariant
and inverse-submixing~(by \cite[Lemma~6]{BBE2010:arxiv} and
\cite[Prop.~3.1]{GimbertKelmendi:IJGT20203}),
the same holds for the conjunction $\mppar$.
For randomized strategies, even with
the improved upper bound in MDPs of~\cref{thm:MR_mppar_MDP},
the construction in \Cref{thm:fmt}
(\cite[Sec.~6.1]{GimbertKelmendi:IJGT20203})
still only yields an exponential upper bound on
the memory of Max strategies in stochastic games.
We show that this is excessive in general.
Optimal randomized Max strategies for $\mppar$ in stochastic games
require, at least and at most, $\numevencolours$
many memory modes, i.e., just linear memory.

\begin{theorem}\label{thm:improved_mppar_games}
Consider a game $\game = \gametuple$, coloring
function $\coloring$ with the highest color $d$, reward function $r$
and objective $\mppar$ for player Max. 
\begin{enumerate}
\item
  If Max can win almost surely from some state $\s$
  then there also exists an almost surely winning randomized Max strategy from $\s$
  that uses at most $k$ public memory modes
  and total bit size $\Ocompl\lrc{{\size{\game}}^{d+c}}$
  where $k = \numevencolours$, $d$ the highest color, and $c$ is a constant
  independent of $\game$.
  \label{itm:upper_bound_mppar}
\item
  The bound on the number of required memory modes for almost surely winning
  randomized Max strategies is tight. There is a family of deterministic
  games $\setcomp{\game_n}{n \ge 2}$ as in~\cref{def:game_n} and \Cref{fig:gn}
  such that
  \begin{itemize}
  \item
    $\size{\game_n}=\Theta(n)$,
    $\game_n$ contains $n$ even colors and every state is almost surely
    winning for Max.
  \item
    For every $n\ge 2$, any Max strategy with $<n$ (private) memory modes is
    worthless, i.e., it cannot guarantee $\mppar$ with any positive probability.  
  \end{itemize}
  \label{itm:lower_bound_mppar}
\end{enumerate}
\end{theorem}

The idea of using randomization to reduce memory complexity is not new and
appears, e.g., in \cite{CDGH2015,DBLP:conf/fossacs/Chatterjee07,DBLP:conf/stacs/Horn09}.
However, it is interesting to note that the Max strategy in SGs still requires
a small amount of memory even in
the presence of randomization and cannot be made memoryless,
unlike in MDPs (\cref{thm:MR_mppar_MDP}).

\begin{table}[ht]
	\centering
	\begin{tabular}{|l|c|c|}
		\hline
		\diagbox[width=8em]{Strategy}{Arena}           & 
		MDPs $\mdp$                                        &
		Stochastic games $\game$
		\\ \hline
		\textbf{Deterministic}                         &
		$\Thetacompl\lrc{\exp\lrc{\size{\mdp}}}$~\cite{Gimbert2011ComputingOS}
		                                               &
		$\Ocompl\lrc{2\!-\!\exp\lrc{\size{\game}}}$~\cite{Gimbert2011ComputingOS,GimbertKelmendi:IJGT20203}
		\\ \hline
		\textbf{Randomized}                            &
		$1$ &
		$\numevencolours = \Thetacompl\lrc{\size{\game}}$
		\\ \hline
	\end{tabular}
	\caption{Worst case number of memory modes required for almost surely winning Max
          strategies for objective $\mppar$ in MDPs and
          stochastic games.
        }
	\label{tab:str_comp_mppar}
\end{table}

\Cref{tab:str_comp_mppar} highlights the strategy complexity of almost surely
winning Max strategies in MDPs and games for the $\mppar$ objective.
Note that while there is no explicit result for the complexity of deterministic
winning strategies, the results from~\cite{Gimbert2011ComputingOS}
and lifting this strategy using \Cref{thm:fmt}
(\cite[Sec.~6.1]{GimbertKelmendi:IJGT20203})
provide a doubly exponential upper bound.
The entries in the randomized row are our contributions. The bit size of these
randomized Max strategies is $\Ocompl\lrc{\size{\mdp}^c}$
and $\Ocompl\lrc{{\size{\game}}^{d+c_1}}$ in MDPs and games
respectively, where $k=\numevencolours$, $d$ is the highest color, $c$ and $c_1$ are constants.

\subparagraph{Upper bound.}
The upper bound in \Cref{thm:improved_mppar_games} Item 1
is shown by induction on the number of colors, with a case distinction
whether the highest color is odd or even.
Max's optimal strategy needs to switch between different modes,
and it uses randomized memory updates with small probabilities
to implement a `probabilistic clock',
in order to stay in one mode for a very long time (in expectation), but not forever;
cf.~\Cref{app:thm41}.

\subparagraph{Outline of the proof.}
We first outline the proof of the upper bound in \Cref{thm:improved_mppar_games},
Item 1. (Full proof in \Cref{app:thm41}.)
It follows the
induction argument of~\cite[Lemma 2,3]{CDGQ:2014} along with strategy
complexity analysis for $\mppar$ instead of $\MP \ge 0\,\cap \,
\Eparity$.
There are two cases, depending on whether the maximum
color in the game is odd or even.

If the maximum color $d$ is odd, then the $\Eparity$ part of $\mppar$
requires that this color must eventually never been seen any more
(except in a null set of the plays).
Using a ranking argument, we show that 
it is possible to partition the state space
into layers $Z_1, \cdots, Z_{\ell}$ such that
\begin{itemize}
\item
  Max can force the win in any one of these layers if the play stays there forever.
\item
  Min cannot infinitely often switch between the layers, except in
  a null set of the plays.
\end{itemize}
The above two facts can be combined to build a winning Max strategy. Interestingly,
this strategy does not use any additional memory, compared to the inductive
case with one less color. It can be seen as a
combination of memoryless attractor strategies on certain states, combined with almost
surely winning strategies in other states in subgames with at least one less
color (using the IH).

If the maximum color $d$ is even, Max tries wherever possible to
reach a state of this color but does so sufficiently infrequently,
so that this does not compromise the satisfaction of the
positive mean-payoff objective. This is achieved by operating in phases with
Max playing either for $\MP > 0 \, \cap \,
\eventually\lrc{\states\lrc{d}}$ or for $\mppar$ in a subgame with at least 1
less color.

\subparagraph{Max's strategy.}
Let $\states$ denote the set of states in the game $\game$, and $X
\eqdef \attr_{\pz}\lrc{\states\lrc{d}}$ the positive attractor of states
with color $d$
in $\game$, $Y \eqdef \states \setminus X$. Since $\game\lrb{Y}$ has at least 1 less
color, by induction hypothesis let $\optzstrat_{Y}$ denote the almost surely
winning strategy for Max in $\game\lrb{Y}$. Let $\optzstrat_{\MP}$ denote the
optimal MD strategy for $\MP > 0$ in $\game$ and $\optzstrat_{\attr}$
denote the positive attractor strategy in states in $X \cap \zstates$.
Then the optimal randomized Max strategy $\optzstrat$ works as follows.
\begin{description}
\item[Phase-1] If the current state is in $X$, play
  the mixed strategy $\eps_0\optzstrat_{\attr} +(1-\eps_0)\optzstrat_{\MP}$.
  I.e., play $\optzstrat_{\attr}$ with some
	      sufficiently small probability $\eps_0 >0$ and play $\optzstrat_{\MP}$ with probability
	      $1-\eps_0$. Else, play according to $\optzstrat_{\MP}$. At every step, there is
	      a sufficiently small chance $\eps_{1} >0$ to stop this phase.
              This is done by changing Max's memory mode with probability $\eps_{1}$.
              So Phase-1 stops eventually almost surely, where the expected
              time to stopping depends on $\eps_{1}$. This serves as
              a makeshift probabilistic clock (since this strategy does not
              use a real clock).
              After this phase is
	      over, if the play is in $X$, restart Phase-1, else move to Phase-2.
\item[Phase-2] Play according to $\optzstrat_{Y}$ while the play is in $Y$. If the
	      play ever moves to $X$, switch to Phase-1.
\end{description}
One has to choose $\eps_0$, $\eps_1>0$ so that the overall strategy is almost
surely winning.
Intuitively, $\eps_0 >0$ is chosen so small that the mean-payoff is strictly 
positive in Phase-1, but this holds only in the long run.
Moreover, while Phase-2 also yields a strictly 
positive mean-payoff in the long run, it might switch back to Phase-1 early
while the accumulated reward is still negative.
(However, this possible negative reward can be bounded, in expectation.)
Therefore $\eps_1>0$ is chosen so small that Phase-1 is played for a very long
time (in expectation). Hence Phase-1 closely approximates its
long run behavior with strictly positive mean-payoff,
and thus it also compensates for any possible negative reward obtained during
a temporary switch to Phase-2.

Also note that this strategy uses randomization in both the
memory updates and the next move. This corresponds to DRR randomized strategy class in
the notation of~\cite[Figure~1.1]{DBLP:conf/concur/MainR22}. Furthermore, the
strategy needs $1$ additional memory mode (compared to strategy $\optzstrat_{Y}$)
in order to know the current phase.
This is necessary, because it needs to play different strategies in $Y$:
$\optzstrat_{\MP}$ in Phase-1 and $\optzstrat_{Y}$ in Phase-2.

See \Cref{app:thm41} for the full proof.

\subparagraph{Lower bound.}
Towards proving the lower bound in \Cref{thm:improved_mppar_games} Item 2,
we construct the following class of games in 
\Cref{fig:gn}.
The idea is that Max needs to remember Min's last move from state $s$
to some state $x$,
since, in order to win, he needs to imitate it from state $t$
by going to state $x-1$.

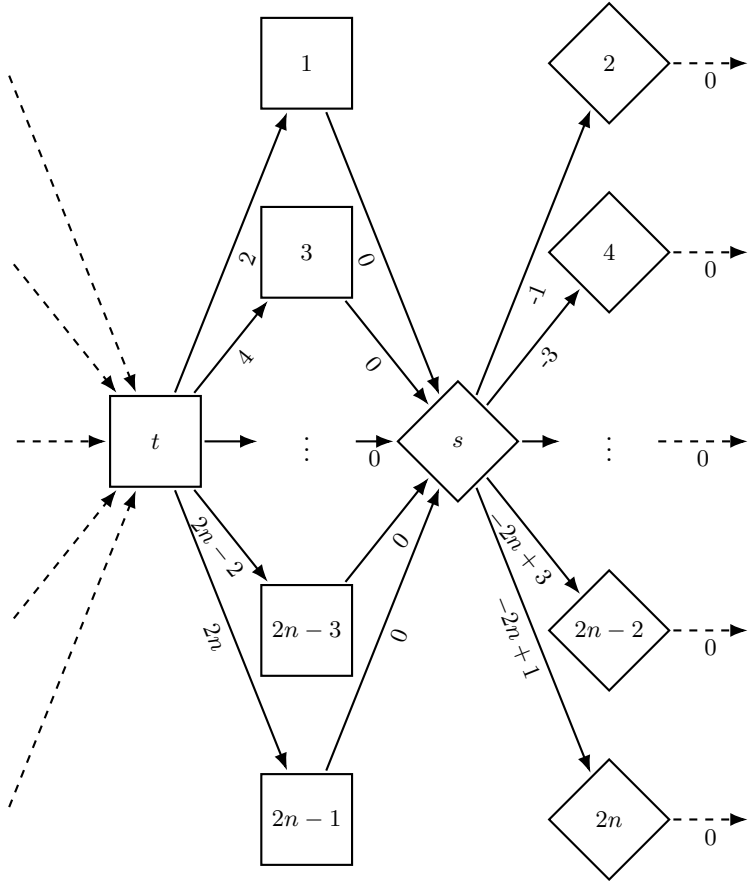
\begin{figure}[ht!]
	\centering
	\begin{tikzpicture}[
    >=Latex,
    max_node/.style={Gmax, minimum size=12mm, font=\small},
    min_node/.style={Gmin, minimum size=16mm, font=\small},
    edge_label/.style={below, sloped, font=\small},
    edge_label2/.style={below, sloped, font=\small},
    invisible_node/.style={opacity=0, minimum size=0pt}
]

\def\xpos{2cm}
\def\ypos{2.5cm}

\node[max_node] (s0) at (0*\xpos,0*\ypos) {$t$};
\foreach \i in {-2,-1,0,1,2}
  \node[invisible_node] (aboves0\i) at (-\xpos,\i*\ypos) {};

\node[max_node] (odd1) at (1*\xpos,2*\ypos) {$1$};
\node[max_node] (odd3) at (1*\xpos,1*\ypos) {$3$};
\node[rectangle, minimum size=12mm] (oddvdots) at (1*\xpos,0*\ypos) {$\vdots$};
\node[max_node] (oddlastbefore) at (1*\xpos,-1*\ypos) {$2n-3$};
\node[max_node] (oddlast) at (1*\xpos,-2*\ypos) {$2n-1$};

\node[min_node] (s1) at (2*\xpos,0*\ypos) {$\s$};

\node[min_node] (even2) at (3*\xpos,2*\ypos) {$2$};
\node[min_node] (even4) at (3*\xpos,1*\ypos) {$4$};
\node[diamond, minimum size = 12mm] (evenvdots) at (3*\xpos,0*\ypos) {$\vdots$};
\node[min_node] (evenlastbefore) at (3*\xpos,-1*\ypos) {$2n-2$};
\node[min_node] (evenlast) at (3*\xpos,-2*\ypos) {$2n$};

\node[invisible_node] (beloweven-2) at (4*\xpos,2*\ypos) {};
\node[invisible_node] (beloweven-1) at (4*\xpos,1*\ypos) {};
\node[invisible_node] (beloweven0) at (4*\xpos,0*\ypos) {};
\node[invisible_node] (beloweven1) at (4*\xpos,-1*\ypos) {};
\node[invisible_node] (beloweven2) at (4*\xpos,-2*\ypos) {};

\foreach \i in {-2,-1,0,1,2}
  \draw[->,dashed] (aboves0\i) edge (s0);

\draw[->] (s0) edge   node[midway, edge_label] {2} (odd1);
\draw[->] (s0) edge   node[midway, edge_label] {4} (odd3);
\draw[->] (s0) edge   node[midway, edge_label] {$2n-2$} (oddlastbefore);
\draw[->] (s0) edge   node[midway, edge_label] {$2n$} (oddlast);
\draw[->] (s0) edge  (oddvdots); 

\draw[->] (odd1) edge  node[midway, edge_label] {0}(s1);
\draw[->] (odd3) edge  node[midway, edge_label] {0}(s1);
\draw[->] (oddlastbefore) edge  node[midway, edge_label] {0}(s1);
\draw[->] (oddlast) edge  node[midway, edge_label] {0}(s1);
\draw[->] (oddvdots) edge node[midway, edge_label] {0}  (s1); 

\draw[->] (s1) edge  node[midway, edge_label2] {\hspace{-1cm}-1}  (even2);
\draw[->] (s1) edge  node[midway, edge_label2] {-3} (even4);
\draw[->] (s1) edge  node[midway, edge_label2] {$-2n+3$} (evenlastbefore);
\draw[->] (s1) edge  node[midway, edge_label2] {$-2n+1$} (evenlast);
\draw[->] (s1) edge (evenvdots); 

\draw[->,dashed] (even2) edge node[midway, edge_label] {0} (beloweven-2);
\draw[->,dashed] (even4) edge node[midway, edge_label] {0} (beloweven-1);
\draw[->,dashed] (evenvdots) edge node[midway, edge_label] {0} (beloweven0);
\draw[->,dashed] (evenlastbefore) edge node[midway, edge_label] {0} (beloweven1);
\draw[->,dashed] (evenlast) edge node[midway, edge_label] {0} (beloweven2) [bend right=15];

\end{tikzpicture}
	\caption{Game $\game_n$ from \Cref{def:game_n} with odd states up to $2n-1$
          and even states up to $2n$.}
	\label{fig:gn}
\end{figure}

The following definition formally defines the class of games in
\Cref{fig:gn}.

\begin{definition}\label{def:game_n}
For each $n \in \Z_+$ let
$\odd^n \eqdef \setcomp{x \in \Z_+}{x \le 2n \,\land\, x \mod 2 = 1}$ and
$\even^n \eqdef \setcomp{x \in \Z_+}{x \le 2n \,\land\, x \mod 2 = 0}$ and
$\game_n \eqdef
	\tuple{\states^n,\tuple{\zstates^n,
			\ostates^n, \rstates^n}, \transition^n,
		\probp^n}$ where
	\begin{enumerate}
		\item $\zstates^n \eqdef \set{t} \uplus \odd^n$,
                  $\ostates^n \eqdef \set{\s} \uplus \even^n$,
                  $\rstates^n \eqdef \emptyset$, consequently $\probp^n$ is
		      trivial.
		\item $\transition^n \eqdef t \times \odd^n \cup \odd^n \times
		      \s \cup \s \times \even^n \cup \even^n \times
		      t$
	\end{enumerate}
        For $n \in \Z_{+}$, state $n$ has color $n$
        and $\coloring\lrc{\s} = \coloring\lrc{t} \eqdef 1$.
	The reward function $r^n$ on the edges is defined as 
	$
        r^n\lrc{\tuple{t,i}} \eqdef i+1, \quad
	r^n\lrc{\tuple{i,\s}} \eqdef 0, \quad
	r^n\lrc{\tuple{\s,j}} \eqdef -j+1, \quad
	r^n\lrc{\tuple{j,t}} \eqdef 0.
      $
\end{definition}

\begin{proof}[Proof of \cref{itm:lower_bound_mppar}(\cref{thm:improved_mppar_games}).]
It is clear from~\cref{def:game_n}
that $\size{\game_n} = \Thetacompl\lrc{n}$ and there are $n$ even colors
in $\game_n$ for every $n \ge 2$.
Min (resp.\ Max) controls only one state $\s$ (resp.\ $t$), where it has a non-trivial choice.
Max can win $\mppar$ surely from state $s$ (and thus from every other state)
by the strategy that, at state $t$, copies Min's most recent observed choice at state $\s$. I.e., if Min's
last step was $\s \to x$ for some even $x$ then Max chooses $t \to (x-1)$.
(This Max strategy requires $n$ memory modes.)
All induced plays trivially satisfy $\Eparity$. Moreover, the total reward between
consecutive visits to state $s$ is always $-(x-1)+x = 1$, and thus
$\MP >0$ is also satisfied surely.

However, we show that all Max strategies with $<n$ private memory modes are worthless.

Towards a contradiction, assume that there exists a Max strategy $\zstrat$
from state $\s$ in $\game_n$ with a set of private memory modes $\memconfset$ where
$|\memconfset| < n$ and
$\inf_\ostrat\Prob[\game_n][\zstrat,\ostrat,\s]{\mppar} > 0$.

First we show, by induction on $k$, the following property $P(k)$
for every $1\le k \le n-1$: 
There exists
a subset $\memconfset_k$ of the memory modes, such that $|\memconfset_k| \ge k$
and for every $\memconf \in \memconfset_k$, the support of $\zstrat[\memconf](t)$
is limited to the $k$ lowest options, i.e., $\subseteq \{1,3,\dots,2k-1\}$.

\textbf{Base case $k=1$}:
Let $\ostrat$ be the Min strategy that always chooses the lowest option $2$.
If, for every $\memconf \in \memconfset$, $\zstrat[\memconf](t)$
includes options \emph{different from} $1$, then
$\Prob[\game_n][\zstrat,\ostrat,\s]{\mppar}
\le \Prob[\game_n][\zstrat,\ostrat,\s]{\Eparity} = 0$, a contradiction.
Hence there must exist at least one memory mode
$\memconf \in \memconfset_1$ such that the support of $\zstrat[\memconf](t)$
is limited to option $1$.

\textbf{Induction step $k-1$ to $k$}:
Let $\ostrat$ be the Min strategy that always chooses the option $2k$.
Consider the long-run behavior of the finite-state Markov chain induced by
$\game_n$, $\zstrat$ and $\ostrat$.
In the runs where Max's memory mode at $t$ is eventually always
contained in $\memconfset_{k-1}$, the total payoff between consecutive
visits to $\s$ is $\le -1$ and therefore $\MP < 0$ and thus these runs
are losing for $\mppar$.
Hence there must exist a non-null set of runs where Max's memory mode
at $t$ is infinitely often in $\memconfset \setminus \memconfset_{k-1}$.
Suppose that for all $\memconf \in \memconfset \setminus \memconfset_{k-1}$
the support of $\zstrat[\memconf](t)$ is \emph{not} limited to the
$k$ lowest options $\{1,\dots,2k-1\}$.
Then $\Prob[\game_n][\zstrat,\ostrat,\s]{\mppar}
\le \Prob[\game_n][\zstrat,\ostrat,\s]{\Eparity} = 0$, a contradiction.
Therefore there exists at least one
$\memconf \in \memconfset \setminus \memconfset_{k-1}$
such that the support of $\zstrat[\memconf](t)$ is limited to the
$k$ lowest options $\{1,\dots,2k-1\}$.
Thus $\memconfset_k \supseteq \{\memconf\} \uplus \memconfset_{k-1}$
and, by induction hypothesis,
$|\memconfset_k| \ge 1 + |\memconfset_{k-1}| \ge 1+(k-1) = k$.
This concludes the induction step.

For $k=n-1$, property $P(n-1)$ yields the required contradiction.
Since $|\memconfset| \le n-1$, we have $\memconfset = \memconfset_{n-1}$.
Thus the support of $\zstrat(t)$ never includes
the highest option $2n-1$. Let $\ostrat$ be the Min strategy that always
chooses the highest option $2n$. Then
$\Prob[\game_n][\zstrat,\ostrat,\s]{\mppar} \le
\Prob[\game_n][\zstrat,\ostrat,\s]{\MP >0} = 0$.
\end{proof}

\section{Counterexamples for Lifting Randomized  Strategies}\label{sec:other-lifting}

A different method to lift strategies from
MDPs to SGs was described in \cite[Theorem 9]{GZ2009}.
Unlike \cite[Theorem 6.1]{GimbertKelmendi:IJGT20203},
it does not require that the objective
is shift-invariant inverse-submixing, but instead has stronger
prerequisites about the strategy complexity in MDPs (resp.\ 1-player games).
Given an objective, if both players have optimal memoryless
\emph{deterministic} strategies in all 
maximizing (resp.\ minimizing) MDPs, then they also have
optimal memoryless deterministic strategies in all SGs.
Under mild assumptions this can be generalized to
finite-memory deterministic strategies
\cite[Theorem 4.1]{Bouyer:LMCS2023}, in stochastic or non-stochastic arenas.

Such results do \emph{not} generalize to \emph{randomized} strategies by
the counterexamples below.

However, first note that there cannot exist any counterexample where the objective
is shift-invariant and submixing.
In that case \cite[Theorem 1.1]{GimbertKelmendi:IJGT20203} implies that Max has optimal
MD strategies in SGs.
This leaves the question whether there are counterexamples that
satisfy other strong properties, e.g., shift-invariant and inverse-submixing.

One counterexample was discussed in \cite[Section 4.4]{Bouyer:LMCS2023}.
For the $\MP=0$ objective, Max (resp.\ Min) have optimal MR (resp.\ MD)
strategies in deterministic 1-player games, but optimal Max strategies require
at least 1 bit of memory in deterministic 2-player games,
even if randomization is allowed.
\cite[Section 4.4]{Bouyer:LMCS2023} does not analyze the strategy complexity
of this objective in MDPs, but it is easy to show that Max (resp.\ Min) has optimal
MR (resp.\ MD) strategies even in MDPs.
It is also easy to extend this example with multiple rewards,
such that optimal Max strategies require more (but still finite)
memory.
However, there remains one downside.
While the $\MP=0$ objective is shift-invariant, it is
neither submixing nor inverse-submixing.
The plays/sequences constructed in the proof of \cite[Lemma 9]{BBE2010:arxiv}
(only in the full version on arXiv) provide counterexamples to either property.

Our lower bounds for the one-dimensional $\mppar$ objective
in \Cref{sec:mp_parity} show a slightly stronger property.
Max (resp.\ Min) have optimal MR (resp.\ MD) strategies even in MDPs,
but Max strategies in deterministic 2-player games can require
arbitrarily many memory modes (equal to the number of even colors),
even if randomization is allowed.
Moreover, the $\mppar$ objective is shift-invariant and inverse-submixing.

The multi-dimensional 
$\MP > \vec{0}$ objective considered in \Cref{sec:lifting1}
is also shift-invariant and inverse-submixing and
Max (resp.\ Min) have optimal MR (resp.\ MD) strategies even in MDPs.
However, optimal Max strategies in deterministic 2-player games can require
an \emph{exponential} (in the dimension) number of memory modes,
even if randomization is allowed.


\begin{figure}[t]
  \centering
\begin{minipage}{.47\textwidth}
  \centering
    \begin{tikzpicture}[node distance= 1.25cm and 2cm,
    acta/.append style={square,inner sep=1}
]

\node (x) at (0,-1) {};  
\node (y) at (0,2) {};  
\node[Gmax] (d) at (0,0) {$s_1$};
\node[Grand,left= of d] (s) {$s_0$};
\node[Gmax,right= of d] (t) {$s_2$};

\draw[->] (s) edge[loop above] node[above]{$\frac{1}{2},+1$} (s);
\draw[->] (s) edge[] node[above]{$\frac{1}{2},+1$} (d);
\draw[->] (d) edge[] node[above]{$0$} (t);
\draw[->] (d) edge[loop above] node[above]{$-1$} (d);
\draw[->] (t) edge[loop above] node[above]{$0$} (t);
\end{tikzpicture}
    \captionof{figure}{
      An MDP where every optimal Max strategy for the objective $W$ requires
      infinite memory. In the random state $s_0$ the successor
      is either $s_0$ or $s_1$, each with probability $1/2$, and the reward is
      $+1$. In the controlled state $s_1$ Max chooses between staying in $s_1$
      with reward $-1$ or going to $s_2$ with reward $0$.
      State $s_2$ is a sink with reward $0$.
      }
\label{fig:W-MDP}
\end{minipage}%
\quad
\begin{minipage}{.47\textwidth}
  \centering
    \begin{tikzpicture}[node distance= 1.25cm and 2cm,
    acta/.append style={square,inner sep=1}
]

\node[Gmax] (d) at (0,0) {$s_1$};
\node (x) at (0,-1) {};  
\node (y) at (0,2) {};  
\node[Gmin,left= of d] (s) {$s_0$};

\draw[->] (s) edge[loop above] node[above]{$-1$}(s);

\draw[->] (s) edge[bend left=50] node[above]{$-1$} (d);

\draw[->] (d) edge[bend left=50] node[below]{$+1$} (s);

\draw[->] (d) edge[loop above] node[above]{$+1$}(d);
\end{tikzpicture}
    \captionof{figure}{
      A deterministic 2-player game where every optimal Max
      strategy for the objective $\obj$ with $X \eqdef \{s_0\}$ requires
      infinite memory. State $s_0$ (resp.\ $s_1$) belongs
      to Min (resp.\ Max) with reward $-1$ (resp.\ $+1$).
      The players choose between staying in the current state or switching.
     \\ \quad
}
\label{fig:O-game}
\end{minipage}
\end{figure}


If one drops the requirement that Min has optimal MD strategies in MDPs,
then there exist even stronger counterexamples where Max requires
infinite memory in deterministic 2-player games.

An example in
\cite[Proposition 3.1.3]{Vandenhove:PhD}
considers the objective $W \eqdef W_1 \cup W_2$, where
$W_1 \eqdef \{c_1c_2 \dots \mid \liminf_{n \to \infty} \sum_{i=1}^n c_i = +\infty\}$ 
and $W_2 \eqdef
\{c_1c_2 \dots \mid \sum_{i=1}^n c_i = 0\ \mbox{for infinitely many $n$}\}$.
Both players have optimal finite-memory deterministic strategies in deterministic
1-player games, but Max needs infinite memory in deterministic
2-player games. However, note that the objectives $W_2$ and $W = W_1 \cup W_2$
are not shift-invariant. Moreover, \cite{Vandenhove:PhD} does not discuss
the strategy complexity of $W$ in MDPs, and
the example in \Cref{fig:W-MDP} shows that optimal Max strategies for $W$
already require infinite memory in MDPs. I.e.,
$W$ is not a counterexample for lifting strategies from MDPs to SGs.

\begin{restatable}{proposition}{WMDP}\label{prop:W-MDP}
Given the MDP in \Cref{fig:W-MDP}, Max has an optimal strategy that attains
the objective $W$ with probability $1$, but every FRR Max strategy is not
optimal.
\end{restatable}
\newcommand{\WMDPproof}{
\begin{proof}
First, since $s_0$ is left almost surely, $W_1$ is a null set under all Max
strategies.

An optimal Max strategy plays as follows.
When $s_1$ is reached with total reward $+n$ (which happens with
probability $2^{-n}$), then it loops $n$ times in $s_1$
and then goes to $s_2$ where the total reward stays $0$ forever.
This satisfies $W_2$ (and thus $W$) with probability $1$.

Now consider an FRR Max strategy with $m$ memory modes.
Since $m$ is finite, there exists at least one memory mode $\memconf$
and numbers $n_2 > n_1 > 0$ such that the two events of entering
$s_1$ with memory mode $\memconf$ and total rewards $n_1$
and $n_2$ each have a nonzero probability.
Thus with nonzero probability either the state $s_2$ is reached with a
total reward $\neq 0$ or $s_2$ is never reached,
and hence $W$ is not satisfied almost surely.
\end{proof}
}
\WMDPproof

We now present a stronger counterexample
where the objective is shift-invariant 
and both Max and Min each have
optimal MR (or alternatively, FDD) strategies
in all MDPs, but optimal Max
strategies still require infinite memory
(even if randomization is allowed) in deterministic 2-player games.

Let 
$\obj_1 \eqdef \{c_1c_2 \dots \mid \limsup_{n \to \infty} \sum_{i=1}^n c_i > -\infty\}$.
Let $X \subseteq \states$ be a subset of the states, e.g., defined via a
coloring function.
Our objective $\obj$ combines $\obj_1$ with B\"uchi and co-B\"uchi objectives
wrt.\ $X$. Let
$
\obj \eqdef \eventually\always X\ \vee\ ((\always\eventually X) \wedge \obj_1)
$.  
The objective $\obj$ is shift-invariant, but neither submixing
nor inverse-submixing.

\begin{restatable}{proposition}{OMRmax}\label{prop:O-MR-max}
Max has optimal MR (and FDD) strategies for $\obj$ in maximizing MDPs.
\end{restatable}
\newcommand{\OMRmaproof}{
\begin{proof}
  Given a maximizing MDP, we can consider the end components wrt.\ an
  optimal strategy. For each winning end component $E$ there are several
  possible cases.

  If $\eventually\always X$ is satisfied almost surely in $E$
  then an MD strategy is sufficient inside $E$, since that
  is a co-B\"uchi objective.

  Otherwise, if it is possible to satisfy
  $((\always\eventually X) \wedge \obj_1)$ almost surely in $E$,
  then there are two cases.
  In the first case, it is possible to satisfy even the stronger
  objective $((\always\eventually X) \wedge \MP>0)$ almost surely in $E$.
  This objective is a special case of $\mppar$ and thus an MR
  (or alternatively FDD) strategy
  is sufficient by \Cref{thm:MR_mppar_MDP}.
  In the second case, it is possible to almost surely satisfy
  $((\always\eventually X) \wedge \obj_1)$ but not
  $((\always\eventually X) \wedge \MP>0)$ inside $E$.
  There must exist a state $x \in X$ inside $E$, such
  that one can satisfy $((\always\eventually x) \wedge \obj_1)$.
  Since $E$ is finite,
  it must be possible to satisfy
  $((\always\eventually x) \wedge \MP=0)$ almost surely.
  Thus, there must exist a strategy from state $x$ such that $x$ is re-visited
  almost surely with an expected total reward $=0$.
  Playing such a strategy repeatedly is also sufficient to satisfy
  $((\always\eventually x) \wedge \obj_1)$ almost surely.
  Thus it suffices to play an MD strategy for the optimal expected total
  reward (paid out upon visiting $x$) inside $E$.

  Hence inside every winning end component an MD strategy
  or an MR (resp.\ FDD) strategy suffices.
  Elsewhere, we can fix an MD strategy that maximizes the chance of
  reaching a winning end component.
  Altogether this yields an optimal MR strategy (resp.\ FDD strategy).
\end{proof}
}
\OMRmaproof

\begin{restatable}{proposition}{OMRmin}\label{prop:O-MR-min}
Min has optimal MR (and FDD) strategies for $\obj$ in minimizing MDPs.
\end{restatable}
\newcommand{\OMRminproof}{
\begin{proof}
  It suffices to show the existence of optimal Max MR (and FDD) strategies
  for the complement objective
  $\overline{\obj}
  =
  \always\eventually\overline{X}\ \wedge\ ((\eventually\always\overline{X}) \vee \overline{\obj_1})
  =
  (\eventually\always\overline{X}) \vee (\always\eventually\overline{X} \wedge \overline{\obj_1})$.

  Like in the proof of \Cref{prop:O-MR-max}, we can consider the strategies
  in winning end components $E$.
  For the co-B\"uchi objective $\eventually\always\overline{X}$, an MD strategy
  in $E$ suffices.
  Otherwise, in finite-state MDPs, we can win 
  $(\always\eventually\overline{X} \wedge \overline{\obj_1})$ almost surely
  iff we can win
  $(\always\eventually\overline{X} \wedge \MP <0)$ almost surely.
  For this an MR
  (or alternatively FDD) strategy
  is sufficient by \Cref{thm:MR_mppar_MDP}.
  
  Hence inside every winning end component an MD strategy
  or an MR (resp.\ FDD) strategy suffices.
  Elsewhere, we can fix an MD strategy that maximizes the chance of
  reaching a winning end component.
  Altogether this yields an optimal MR strategy (resp.\ FDD strategy).
\end{proof}
}
\OMRminproof

However, in deterministic 2-player games, optimal Max strategies
for $\obj$ require infinite memory.

\begin{restatable}{proposition}{Oinfgames}\label{prop:O-inf-games}
  Consider the deterministic 2-player game $\game$ from \Cref{fig:O-game},
  and objective $\obj$ with $X \eqdef \{s_0\}$.
  Max can win objective $\obj$ surely, but every FRR Max strategy
  cannot guarantee any positive probability for $\obj$.
\end{restatable}
\newcommand{\Oinfgamesproof}{
\begin{proof}
Towards the first part, we define an optimal Max strategy
that keeps track of the total reward and plays
as follows.
Whenever $s_1$ is reached with some negative total reward $-n$,
Max loops $n-1$ times at $s_1$
and then goes back to $s_0$, which makes the total reward $0$.
This ensures objective $\obj$ against any Min strategy as follows.
In plays where Min eventually stays in $s_0$ forever
we have $\eventually\always X$ and thus $\obj$.
Otherwise, $s_0$ and $s_1$ are both visited infinitely often and
the $\limsup$ of the total reward is $0$, and thus $\obj$ holds.

Towards the second part, we
consider an FRR Max strategy $\zstrat$ with $m$ private memory modes,
starting in memory mode $\memconf_0$.
We will construct a Min strategy $\ostrat$ such that
$\probm^{\game}_{\zstrat,\ostrat,s_0}(\obj)=0$.

For every memory mode $\memconf$ let $p(\memconf)$ be probability
that, in state $s_1$ and memory mode $\memconf$, in the next step
$\zstrat$ returns to $s_0$.
Let $p \eqdef \min\{p(\memconf) \mid p(\memconf)>0\} > 0$.
Let $\ostrat$ be the MR Min strategy that in $s_0$ goes to
$s_1$ with probability $p/2$ and to $s_0$ with probability $1-p/2$.
Fixing the strategies $\zstrat,\ostrat$ yields a Markov chain with
$2m$ states and initial state $(s_0,\memconf_0)$.

Let $E$ be the event that $s_1$ is visited when $\zstrat$ is in some
memory mode $\memconf$ with $p(\memconf)=0$.

Conditioned under $E$, the properties $\eventually\always X$
and $\always\eventually X$ 
are not satisfied and thus $\obj$ is not satisfied, i.e.,
$\probm^{\game}_{\zstrat,\ostrat,s_0}(\obj \cap E) =0$.
If $\overline{E}$ is a null set then this shows the required property.

Otherwise, conditioned under $\overline{E}$, both states $s_0$ and $s_1$
are visited infinitely often almost surely, since $p > p/2 > 0$.
In the steady state, the probability $\alpha$ of being in $s_0$
satisfies $\alpha \ge \alpha (1-p/2) + (1-\alpha)p$,
and therefore $\alpha \ge 2/3$.
Hence the (conditional under $\overline{E}$)
expected mean payoff is $\le \alpha(-1)+(1-\alpha) \le -1/3$.
It follows that $\probm^{\game}_{\zstrat,\ostrat,s_0}(\MP < 0 \mid \overline{E})=1$
and thus $\probm^{\game}_{\zstrat,\ostrat,s_0}(\obj_1 \mid \overline{E})=0$.
Moreover have have
$\probm^{\game}_{\zstrat,\ostrat,s_0}(\eventually\always X)=0$, since $p/2 >0$
and thus
$\probm^{\game}_{\zstrat,\ostrat,s_0}(\obj \cap \overline{E}) = 0$.

Finally, $\probm^{\game}_{\zstrat,\ostrat,s_0}(\obj)
=
\probm^{\game}_{\zstrat,\ostrat,s_0}(\obj \cap E) +
\probm^{\game}_{\zstrat,\ostrat,s_0}(\obj \cap \overline{E})
=0$.
\end{proof}
}

\Oinfgamesproof

\section{Conclusion and Open Questions}\label{sec:conclusion}

While finite-memory Max strategies for shift-invariant inverse-submixing
objectives can be lifted from MDPs to 2-player stochastic games
by \cite[Theorem 6.1]{GimbertKelmendi:IJGT20203},
this yields doubly exponentially many extra memory modes in general and 
exponentially many extra memory modes when lifting memoryless randomized
strategies.
In the latter case, this exponential extra memory cannot be avoided in general
(\Cref{thm:trigger-lower}).

An open question concerns the exact trade-off between memory and attainment,
i.e., how good can randomized strategies with a sub-optimal number of memory
modes be in the worst case, relative to strategies with sufficient memory.

The mean-payoff-parity objective $\mppar$ in \Cref{sec:mp_parity}
provides an interesting example where randomization in strategies
drastically reduces the amount of memory required (from doubly 
exponential to linear), but does not eliminate the need for memory
entirely.

An open question is whether low-memory strategies for $\mppar$ require
randomized memory updates, or whether just randomized actions and
deterministic memory updates would suffice.


\newpage
\bibliography{conferences,journals,ref1}

\newpage
\appendix
\section{Appendix for~\cref{sec:prelim}}\label{app:prelim}

\subparagraph{Reachability.} We use the syntax and semantics of the LTL
operators~\cite{CGP:book}
\mbox{$\eventually$ (eventually)} and $\always$ (always)
to specify some conditions on plays.
A \emph{reachability objective} is defined by a set of target states
$\reachset \subseteq \states$. A play $\play = \s_0s_1 \ldots$
belongs to $\eventually\,\reachset$ iff $\exists i \in \nat\, \s_i
\in \reachset$.
Similarly, $\play$ belongs to $\eventually^{\le n} \reachset$
(resp.\ $\eventually^{\ge n} \reachset$) iff
$\exists i \le n$ (resp.\ $i \ge n$) such that $\s_i \in \reachset$.
Dually, the \emph{safety} objective
$\always\,\reachset$ consists of all plays
which never leave $\reachset$. We have $\always\,\reachset =
\neg\eventually\neg\reachset$.

\subparagraph{Parity.} A parity objective is defined via a bounded function
$\coloring\!: \states \to \nat$ that assigns non-negative priorities
(aka colors) to states. Given an infinite play
$\play = \s_0s_1 \ldots$, let $\text{Inf}(\play)$
denote the set of numbers that occur infinitely often in the sequence
$\coloring(\s_0)\coloring(\s_1)\ldots$.
A play $\play$ satisfies \textit{even parity}
w.r.t.\ $\coloring$ iff the maximum \footnote{The maximum exists since
$\coloring$ is bounded, even in the general case where $\states$ is not
finite. However, in this paper, we only consider finite sets of states $\states$.}
of $\text{Inf}(\play)$ is even.
Otherwise, $\play$ satisfies \textit{odd parity}.
The objective even parity is denoted by $\Eparity(\coloring)$
and odd parity is denoted by $\Oparity(\coloring)$.
Most of the time, we implicitly assume that the coloring function is known
and just write $\Eparity$ and $\Oparity$.
Observe that, given any coloring $\coloring$, we have
$
\complementof{\Eparity} = \Oparity$ and
$
\Oparity(\coloring) = \Eparity(\coloring + 1)
$
where $\coloring + 1$ is the function which adds $1$ to the color of every
state. This justifies to consider only one of the even/odd parity objectives,
but, for the sake of clarity, we distinguish these objectives
wherever necessary. For any subset $T$ and color $c$, we denote by $T\lrc{c}$,
states in $T$ with color $c$.

\subparagraph{Attractors, traps and subgames.}
A non-empty set $H \subseteq \states$ defines a \emph{subgame} iff for all $\s
 \, \in H \, \cap \, \rstates$, $\successors{\s} \subseteq H$ and every state
in $H$ has at least one successor in $H$. The resulting subgame which is
well-defined is denoted by $\game\lrb{H}$.
For a set $T \subseteq \states$, the \emph{positive attractor} for player
$\px \in \{\pz,\po\}$
in game $\game$, denoted by $\attr_{\px}\lrc{T,\game}$ is the set of states 
$\s$ where $\px$ has a strategy to ensure that $\eventually T$ is satisfied
with positive probability when starting from $\s$. For any given $T$, the
positive attractor set and a uniform MD strategy $\xstrat_{\attr}$ which
satisfies this can be computed in polynomial time. A set $U \subseteq \states$
is a \emph{trap} for Min iff for every
Min/random state in $U$, the successors are always in $U$ and there is at
least one successor in $U$ for every Max state.
One can similarly define a trap for Max. Note that by definition
a trap for either player is a subgame. Also, for any set $T$, $\states 
 \setminus \attr_{\px}\lrc{T,\game}$ is a trap for $\px$ if it is non-empty.

\section{Appendix for~\texorpdfstring{\Cref{sec:mp_parity}}{MeanPayoff-Parity}}

\ignore{
\subsection{Proof of \Cref{thm:MR_mppar_MDP}}

\MRmpparMDP*

\begin{proof}

\end{proof}

}

\subsection{Proof of \Cref{thm:improved_mppar_games}, Item 1}\label{app:thm41}

\begin{proof}[Proof~of~\cref{itm:upper_bound_mppar}(~\cref{thm:improved_mppar_games}).]
	The proof is by strong induction on the maximum color $d$.\\
	Let $P_d$ denote the statement as given
	by~\cref{itm:upper_bound_mppar}. We then prove the following.
	\begin{align}
		P_0                                             & \lrc{\text{if minimum color is 0}} \label{eqn:base-zero} \\
		P_1                                             & \lrc{\text{if minimum color is 1}} \label{eqn:base-one} \\
		k \ge 0\; \lrc{\forall\;s\,\le\, 2k.\; P_{s}}  & \implies P_{2k+1} \label{eqn:induction-odd}                   \\
		k \ge 0\; \lrc{\forall\;s\,\le\,2k+1.\; P_{s}} & \implies P_{2k+2} \label{eqn:induction-even}
	\end{align}
	\textbf{Base case}: 
	When there is only one color, $\Eparity$ is trivial(~\cref{eqn:base-zero}) or never satisfied(~\cref{eqn:base-one})
	and hence $\mppar$ is equivalent to $\MP > 0$ for which MD
	strategies exist \cite[Prop.~7]{Brazdil2010}.
        Since $d$ is either $0$ or $1$, either the number of even colors
	is $1$ or no almost surely winning strategy exists for Max.\\
	\textbf{Induction step}: 
	Assume that the statement is
	true for all $k < d$. W.l.o.g.\ one can assume that every state $\s$ is
	almost surely winning as otherwise it is possible to consider a 
	subgame with states that satisfy this condition.
	To show that statement holds for $d$, we split into two cases
	depending on whether $d$ is even or odd.
        The following two claims are proven below.

	\begin{restatable}[Maximum Color Odd]{claim}{maxcolourodd}\label{claim:max_colour_odd}
		Let $d = 2k+1$ be odd. Then there is an almost surely winning
                Max strategy
		$\optzstrat$ which can be constructed from a finite number of
		$\optzstrat_i$ which are finite-memory almost surely winning strategies in
		subgames with maximum color $< d$. Moreover, the number of memory
		modes in $\optzstrat$ is equal to the maximum number of memory modes in
		$\optzstrat_i$. If all $\optzstrat_i$ have rational probabilities, then
		$\bits(\optzstrat) = \Ocompl(\sum_i \bits(\optzstrat_i)) = \Ocompl(
		\card{\states} \bits(\optzstrat_i) )$
	\end{restatable}

	From the above lemma, observe that every state $\s$ in the almost
	sure winning region can be uniquely identified as belonging to one of
	these sets.
	\begin{enumerate}
		\item $\s \in Z_i \setminus R_i$ for some $i$
		      \label{itm:attractor_state}
		\item $\s \in R_i$
	\end{enumerate}

	The strategy $\optzstrat$ can be constructed as follows.
	\begin{itemize}
		\item Since each $R_i$ has at least one less color, obtain the strategies
		      $\optzstrat_i$ by induction.
		\item In each $Z_i$, find the MD attractor strategy $\optzstrat_{\attr,i}$
		\item If the play is currently in $Z_i$, then
		      \begin{itemize}
			      \item If the current state is of type~\cref{itm:attractor_state}, play according to
			            $\optzstrat_{\attr,i}$
			      \item Else if we land into $R_i$, begin playing according to $\optzstrat_i$ and
			            continue this until the play exits $R_i$.
		      \end{itemize}
		\item If at any point, the play exits from $Z_i$, reset and start again
	\end{itemize}

	\begin{restatable}[Maximum Color Even]{claim}{maxcoloureven} \label{claim:max_colour_even}
		Let $d > 0$ be even and assume there are $k$ even colors. Define $X \eqdef
		\attr_{\pz}\lrc{\states\lrc{d},\game}$ and $Y
		\eqdef
		\states \setminus X$.
		\begin{enumerate}
			\item $\states \subseteq \AS[\game][\pz]{\mppar}$ if
			      and only if
			      \begin{itemize}
				      \item $\states \subseteq
				            \AS[\game][\pz]{\MP > 0}$ and
				      \item $Y \subseteq
				            \AS[\game\lrb{Y}][\pz]{\mppar}$
			      \end{itemize} \label{itm:equivalence_even}
			\item Let $\optzstrat_{Y}$ denote a Max
			      strategy in $\game\lrb{Y}$, which has finite memory and
				  is almost surely winning from
				  every state in $Y$. Then an almost surely winning
                  Max strategy $\optzstrat$ in $\game$
			      can be constructed such that $\size{\optzstrat} = 1 +
				  \size{\optzstrat_{Y}}$. Furthermore, if all the
				  probabilities in $\optzstrat_{Y}$ are rational, then 
				  $\bits(\optzstrat) = \Ocompl\lrc{\size{\game}^{c_1} 
				  \size{\optzstrat_Y} + 
			 	   \bits(\optzstrat_Y) \size{\game}}$ where
				  $c_1$ is a fixed constant independent of $\game$.
				  \label{itm:strategy_even}
		\end{enumerate}
	\end{restatable}

	In \cref{claim:max_colour_odd}, $\size{\optzstrat}$ 
    is the maximum of the $\size{\optzstrat_i}$ which, by induction
	hypothesis, is at most $\numevencolours$.
    Since, in this case, the maximum color is odd,
    the number of even colors in the subgames is at most the same as in the
	whole game. Hence $\size{\optzstrat} = \max_i \size{\optzstrat_i}$. By
	induction hypothesis $\bits(\optzstrat_i) = \Ocompl\lrc{
	{\size{\game}}^{d-1+c}}$, and from~\cref{claim:max_colour_odd}, we get
	$\bits(\optzstrat) = \Ocompl(\card{\states} \bits(\optzstrat_i)) = 
	\Ocompl(\size{\game} \bits(\optzstrat_i)) = \Ocompl(\size{\game}^{d+c})$.
    This shows~\cref{eqn:induction-odd}.

	When the maximum color $d$ is even and $k = \numevencolours$,
	\cref{claim:max_colour_even} implies that the
	number of even colors in $Y$ is at most $k-1$.
        Since $\size{\optzstrat} = 1 + \size{\optzstrat_Y}$, we have proven the induction step
	for the number of memory modes. For the bit size, by the induction hypothesis
	$\bits(\optzstrat_Y) = \Ocompl\lrc{\size{\game}^{d-1+c}}$. 
	From~\cref{claim:max_colour_even},
	\begin{align*}
		\bits(\optzstrat) &= \Ocompl\lrc{\size{\game}^{c_1} \size{\optzstrat_Y}
		+ \bits(\optzstrat_Y) \size{\game}} \\
		&= \Ocompl\lrc{\size{\game}^{c_1+1} + \size{\game}^{d+c}} \\
		&= \Ocompl\lrc{\size{\game}^{d+c}}
	\end{align*}
	This shows~\cref{eqn:induction-even}.
\end{proof}

\subsubsection{Proofs of claims used above}

The following auxiliary claim is used in the proofs of
\Cref{claim:max_colour_odd}
and \Cref{claim:max_colour_even}.

\begin{claim}[Positive Attractor property]\label{claim:attr_properties}
		Consider a game $\game = \gametuple$, a subset of the states
                $H \subseteq S$,
                and a finite-memory Max strategy $\optzstrat$ that satisfies the following
                condition:
                $\forall \s \in \attr_{\pz}\lrc{H}\, \forall \memconf.\inf_{\ostrat}
                \Prob[][\optzstrat[\memconf],\ostrat,\s]{\eventually\, H} >0$.
                Then the following two properties hold. ($\symdiff$ denotes
                the symmetric difference between sets.)
		\begin{enumerate}
			\item $\Prob[][\optzstrat,\ostrat,\s_0]{\always\,
				      \eventually\,
				      H\;\symdiff\; \always\,
				      \eventually\, \lrc{\attr_{\pz}H} } =
			      0$.\label{itm:attractor_target_equiv}
			\item Furthermore, if $H$ is a Min-trap and
                          $\optzstrat$ always remains inside 
			      $\game\lrb{H}$ after entering H, then
			      $\Prob[][\optzstrat,\ostrat,\s_0]{\eventually\, \always\, H\;\symdiff\;
				      \always\, \eventually\, H} = 0$.\label{itm:target_target_equiv}
		\end{enumerate}
	\end{claim}
	\begin{claimproof}[Proof of~\cref{claim:attr_properties}]
          Towards \cref{itm:attractor_target_equiv}, we use the facts that
          $\attr_{\pz}\lrc{H} \subseteq S$
          is finite and $\optzstrat$ is finite-memory to obtain that
          $p \eqdef \min_{\s \in \attr_{\pz}\lrc{H}} \min_\memconf\ \inf_{\ostrat}
          \Prob[][\optzstrat[\memconf],\ostrat,\s]{\eventually\, H} >0$.
          Hence $\Prob[][\optzstrat,\ostrat,\s_0]{\always\, \eventually\,
            H\;\symdiff\; \always\, \eventually\, \lrc{\attr_{\pz}H} } =
          \Prob[][\optzstrat,\ostrat,\s_0]{\always\, \eventually\,
            \lrc{\attr_{\pz}H} \setminus \always\, \eventually\, H} \le
          (1-p)^\infty = 0$.

          \cref{itm:target_target_equiv} follows directly from the assumptions. 
	\end{claimproof}
	It is easy to see that the above claim also holds true from the perspective
	of Min, i.e., when Max and Min roles are reversed.

    \maxcolourodd*
    \begin{claimproof}[Proof.
			of~\cref{claim:max_colour_odd}]
		For Max to win $\mppar$ almost surely, it has to eventually not see any
		of the states from $\states\lrc{d}$ any more. To achieve this, we `rank' the
		states in $\states\lrc{d}$ so that it is not possible for Min to force
		a move from a lower ranked state to a higher ranked one.
                Ultimately,
		this implies that eventually almost surely we are always inside states with the
		same rank. If we show that Max can win here, then we are done.
		Formalizing this intuition, we need to show that
		\begin{enumerate}
			\item There exists a partition $\set{Z_i}_{1 \le
				      i \le \ell}$ of
                                    $\states$ and non-empty sets $R_i, U_i$ for
                                    $1 \le i \le \ell$ 
                                    where $U_1 = \states$, and $U_{\ell+1} =
                                    \emptyset$ such that
			      \begin{enumerate}
				      \item $R_i \subseteq U_i \setminus
				            U_i\lrc{d} \neq \emptyset$
				            is a trap for Min in $\game\lrb{U_i}$
				            and $R_i
				            \subseteq
				            \AS[\game\lrb{R_i}][\pz]{\mppar}$ with corresponding
				            winning strategy $\optzstrat_{i}$
				      \item $Z_i = \attr_{\pz}\lrc{R_i,
					            \game\lrb{U_i}}$
				      \item $U_{i+1} = U_i \setminus Z_i$
			      \end{enumerate} \label{itm:partition_characterization}
		\end{enumerate}
		The states in $Z_i$ as defined above for $1 \le i \le \ell$ can be
		thought of as having rank $i$. Note that finding a non-empty $R_i$ in
		$U_i$ fixes $Z_i$ and $U_{i+1}$. The above characterization can be
		seen as inductively defining the sets $Z_i$ as long as one can
		find the non-empty set $R_i$. The induction stops when $U_{i+1}$ is
		empty. The $R_i$ is not unique and different
		$R_i$'s thus give rise to different partitions. The claim holds
		for any partition satisfying the above conditions.
		To explicitly find one such non-empty $R_i$, we can use the following claim.

		\begin{claim} \label{claim:nonempty_winning_set}
			Let the maximum color $d$ be odd and all states $\in \;
			 \AS[][\pz]{\mppar}$. Let
			$\emptyset \subset U \subseteq \states$ be a trap for player
			Max. Define
			$X_U \eqdef
			 \attr_{\po}\lrc{U\lrc{d},\game\lrb{U}}$ and
			$Y_U \eqdef U \setminus X_U$. Then
			\[
				R_U \eqdef
				\AS[\game\lrb{Y_U}][\pz]{\mppar} \neq \emptyset.
			\]
		\end{claim}
		\begin{claimproof}[Proof of~\cref{claim:nonempty_winning_set}]
            Towards a contradiction, assume $R_U = \emptyset$.
            We'll show that this implies that there is a Min strategy
			$\optostrat$ always staying in $U$ such that for all states $\s\,
			\in \, U$ and all strategies $\zstrat$ of Max
			\[
				\Prob[\game\lrb{U}][\zstrat,\optostrat,\s]{\mppar} < 1.
			\]
			$R_U = \emptyset$ implies there is a Min strategy $\ostrat = 
			\memstrattuple$ in $\oallstrats{\game\lrb{Y_U}}$ such that
			for all $\s \, \in \, Y_U$ and all strategies $\zstrat$ which
			always stay in $Y_U$, 
			$\Prob[\game\lrb{Y_U}][\zstrat,\ostrat\lrb{\initmem},\s]{\mppar}
			< 1.$
			Let $E \eqdef \eventually\,\always\, Y_U$ and
			$F \eqdef \complementof{E}$ in $\game\lrb{U}$. Then,
			$F = \always\,\eventually\,X_U$. By definition of $X_U$, there
			is a MD Min strategy $\ostrat_{\attr,U}$ from every Min state in
			$X_U$. Consider the following Min strategy
			$\optostrat \eqdef
			\tuple{\memconfset^\prime,\memup^\prime, \memsuc^\prime}$ where
			\begin{itemize}
				\item $\memconfset^\prime \eqdef \memconfset$
				\item $\memsuc^\prime\lrc{\memconf,\s} \eqdef
				      \begin{cases}

					      \ostrat_{\attr,U}\lrc{\s} & \s\,\in\, X_U \\
					      \memsuc\lrc{\memconf,\s}
					                                & \s\,\in\, Y_U
				      \end{cases}$
				\item
				      $\memup^\prime\lrc{\memconf,\tuple{\s,\s'}} \eqdef \begin{cases}
					      \initmem
					                                          & \s\text{ or }\s'\, \in\, X_U \\

					      \memup\lrc{\memconf,\tuple{\s,\s'}} & \text{otherwise}
				      \end{cases}$
			\end{itemize}
			Observe that $\optostrat\lrb{\initmem}$ never leaves $U$. Consider
			any Max strategy $\zstrat$, start state $\s \in U$. Then
			$\Prob[\game][\zstrat,\optostrat\lrb{\initmem},\s]{\mppar}$
			\begin{align*}
					& = \Prob[\game][\zstrat,\optostrat\lrb{\initmem},\s]{\mppar\, \cap \, E} 
					& +
					&\Prob[\game][\zstrat,\optostrat\lrb{\initmem},\s]{\mppar\, \cap \, F}
				\\
					& = \Prob[\game][\zstrat,\optostrat\lrb{\initmem},\s]{E} 
					\cdot \underbrace{\Prob[\game\lrb{Y_U}][\zstrat,\optostrat\lrb{\initmem},\alpha]{\mppar}}_{< 1}
					& +
					&\Prob[\game][\zstrat,\optostrat\lrb{\initmem},\s]{\mppar\, \cap \, F} 
				\\
				 & <  \Prob[\game][\zstrat,\optostrat\lrb{\initmem},\s]{E}
				 & + 
				 &\underbrace{\Prob[\game][\zstrat,\optostrat\lrb{\initmem},\s]{\Eparity \, \cap \, \always\,\eventually\, U\lrc{d} }}_{\text{by}~\cref{claim:attr_properties}}
				\\
				&< \Prob[\game][\zstrat,\optostrat\lrb{\initmem},\s]{E} & + & \underbrace{0}_{d \text{ is the max color and odd}} \\
				&< \Prob[\game][\zstrat,\optostrat\lrb{\initmem},\s]{E} & & \\
				&< 1 & &
			\end{align*}
			which implies $U \not\subseteq \AS[][\pz]{\mppar}$, a contradiction.
		\end{claimproof}

		Since $U_1 = \states$ itself is a trap for
		Max,~\cref{claim:nonempty_winning_set} implies existence of a non-empty set
		$R_1$ such that $R_1 \subseteq U_1 \setminus U_1\lrc{d}$ and $R_1 \subseteq
		\AS[\game\lrb{R_1}][\pz]{\mppar}$ (true since a winning strategy never
		leaves the almost sure set of states). Define $Z_1 \eqdef
		\attr_{\pz}\lrc{R_1,\game\lrb{U_1}}$ and let $U_2 \eqdef U_1 \setminus Z_1$.
		Observe that $U_2$ is a strict subset of $U_1$ and is once again a trap for
		Max. If the highest color in $U_2$ is $d$, one can again
		apply~\cref{claim:nonempty_winning_set} or else trivially take $R_2 \eqdef U_2$
		and the procedure stops.
		
		For our purposes, any partition which
		satisfies~\cref{itm:partition_characterization} suffices for the proof.
		Observe that
		$U_i = \bigcup_{k=i}^{\ell} Z_k$ and $U_i \supset U_{i+1}$.
		Also
		$R_i \subseteq Z_i$ and $R_i \subseteq U_i \setminus U_i(d)$
		implies
		$R_i \subseteq Z_i \setminus Z_i(d)$. Furthermore, given
		$\optzstrat_i\lrb{\initmem} = \tuple{\memconfset,\memsuc_i,
			\memup_i}$ which is winning from every state in $R_i$
		and the uniform MD attractor strategies $\optzstrat_{\attr,i}$
		to
		$R_i$ in $Z_i$, we construct a strategy
		$\optzstrat\lrb{\initmem} \eqdef \tuple{\memconfset,\memsuc,
			\memup}$ as follows.
		\begin{align*}
			\memsuc\lrc{\memconf, \s}             & \eqdef
			\begin{cases}
				\optzstrat_{\attr,i}\lrc{\s} & \s\,\in\,Z_i \setminus R_i
				\text{ for some } 1 \le i \le \ell                            \\
				\memsuc_i\lrc{\memconf,\s}   & \s\,\in\,R_i \text{ for some } 1
				\le i \le \ell
			\end{cases}                     \\
			\memup\lrc{\memconf, \tuple{\s, \s'}} & \eqdef
			\begin{cases}
				\initmem                            & \s\,\in\,Z_i \setminus R_i,\; \s'\,\in\,Z_i
				\text{ for some } 1 \le i \le \ell                                                \\
				\memup_i\lrc{\memconf,\tuple{\s,\s'}} & \s,\; \s'\,\in\,R_i
				\text{ for some } 1 \le i \le \ell                                                \\
				\initmem                            & \s\,\in\,Z_i,\; \s'\,\in\,Z_j \; i \neq j
			\end{cases}
		\end{align*}
		We assumed that all the strategies $\optzstrat_i$ share the same set of memory
		configurations $\memconfset$ and start in the same memory configuration
		$\initmem$. If they are different, one can take $\memconfset$ such that
		$\card{\memconfset} = \max_i\lrc{\card{\memconfset_i}}$ and simple renaming of
		the configurations in the strategies so that every start configuration is
		renamed to $\initmem$. From the definition, it is clear that $\size{
		\optzstrat} = \max\limits_i\size{\optzstrat_i}$ and $\bits(\optzstrat)
		= \sum\limits_i \bits(\memsuc_i) + \bits(\memup_i) = \Ocompl(
		\card{\states} \bits(\optzstrat_i) )$. This takes care of the
		complexity part of the claim. We now show that $\optzstrat[\initmem]$ is almost surely
		winning from every state $\s$.

		Let $\obj_i$ denote $\eventually\,\always\,R_i$ for each $1\le i
		\le \ell$ and consider the event $\mppar\;\cap\;\obj_i$. Fix some
		arbitrary strategy $\ostrat$ for Min and assume
		$\Prob[\game][\optzstrat\lrb{\initmem},\ostrat,\s]{\obj_i} > 0$ as otherwise
		the conclusion obtained below is trivial. For every play $\play \in \obj_i$,
		let the random variable $T_i$ denote the hitting time of a state which
		satisfies $\always\,R_i$ and $X_{T_i}$ denote the state in which one
		enters $R_i$. Also, let $\alpha_i$ denote the distribution of
		$X_{T_i}$ and $\ostrat_i \, \in \oallstrats{\game\lrb{R_i}}$ denote the
		Min strategy which simulates the play until $\always\,R_i$ in $\game$
		and then plays according to $\ostrat$. Then
		$$ \Prob[\game][\optzstrat\lrb{\initmem},\ostrat,\s]{\mppar\;\cap\;\obj_i}
			= \Prob[\game][\optzstrat\lrb{\initmem},\ostrat,\s]{\mppar
				\given \obj_i} \cdot
			 \Prob[\game][\optzstrat\lrb{\initmem},\ostrat,\s]{\obj_i}$$
		\begin{align*}	
			& = \sum_{\s' \in
				R_i}\Prob[\game][\optzstrat\lrb{\initmem},\ostrat,\s]{\mppar
				\given X_{T_i} = \s'} 
			 \cdot \Prob[\game][\optzstrat\lrb{\initmem},\ostrat,\s]{X_{T_i}=\s' \given \obj_i} 
			 \cdot \Prob[\game][\optzstrat\lrb{\initmem},\ostrat,\s]{\obj_i}      \\
			& = \Prob[\game\lrb{R_i}][\optzstrat_i\lrb{\initmem},\ostrat_i,\alpha_i]{\mppar} \cdot
			 \Prob[\game][\optzstrat\lrb{\initmem},\ostrat,\s]{\obj_i}      \\
			& = \Prob[\game][\optzstrat\lrb{\initmem},\ostrat,\s]{\obj_i}
		\end{align*}

		From the above discussion, one can see that
		$\Prob[\game][\optzstrat\lrb{\initmem},\ostrat,\s]{\mppar\;\cap\;\bigcup_i\obj_i} = \Prob[\game][\optzstrat\lrb{\initmem},\ostrat,\s]{\bigcup_i\obj_i}$.
		If we show the latter event occurs with probability $1$,
		then we are done with the converse since
		$$ \Prob[\game][\optzstrat\lrb{\initmem},\ostrat,\s]{\mppar}
			\ge
			\Prob[\game][\optzstrat\lrb{\initmem},\ostrat,\s]{\mppar\;\cap\;\bigcup_i\obj_i}
			= 1.$$
		Consider the complement objective and its probability
		$\Prob[\game][\optzstrat\lrb{\initmem},\ostrat,\s]{\complementof{\bigcup_i\obj_i}}$
		which can equivalently be written as
		$\Prob[\game][\optzstrat\lrb{\initmem},\ostrat,\s]{\bigcap_i\complementof{\obj_i}}$.
		$Z_i$ is the attractor for $R_i$ within $\game\lrb{U_i}$ where
		$R_i$ is also a Min-Trap. $\optzstrat\lrb{\initmem}$ when in 
		$\game\lrb{U_i}$ satisfies
		the hypothesis given in~\cref{claim:attr_properties}. Therefore
		for every $i$,
		\begin{eqnarray}
			\Prob[\game][\optzstrat\lrb{\initmem},\ostrat,\s]{\lrc{\always\,
					\eventually\,  Z_i\;\symdiff\; \always\, \eventually\, R_i}
				\given \eventually\, \always\, U_i} &=& 0
			\label{eqn:restriced_equiv_Z_i_R_i} \\
			\Prob[\game][\optzstrat\lrb{\initmem},\ostrat,\s]{\lrc{\always\,
					\eventually\,  R_i\;\symdiff\; \obj_i} \given \eventually\,
				\always\, U_i} &=& 0 \label{eqn:restriced_equiv_R_i_R_i}
		\end{eqnarray}

		From~\eqref{eqn:restriced_equiv_Z_i_R_i},~\eqref{eqn:restriced_equiv_R_i_R_i},
		one can then conclude
		\begin{align*}
			\Prob[\game][\optzstrat\lrb{\initmem},\ostrat,\s]{\eventually\,
			\always\, U_i\;\cap\; \complementof{\obj_i}}
			 & =
			\Prob[\game][\optzstrat\lrb{\initmem},\ostrat,\s]{\eventually\,
				\always\, U_i\;\cap\; \complementof{\always\, \eventually\,  R_i}} \\
			 & =
			\Prob[\game][\optzstrat\lrb{\initmem},\ostrat,\s]{\eventually\,
				\always\, U_i\;\cap\; \complementof{\always\, \eventually\,  Z_i}} \\
			 & =
			\Prob[\game][\optzstrat\lrb{\initmem},\ostrat,\s]{\eventually\,
				\always\, U_i\;\cap\; \eventually\, \always\,\complementof{Z_i}}   \\
			 & =
			\Prob[\game][\optzstrat\lrb{\initmem},\ostrat,\s]{\eventually\,
				\always\, \lrc{U_i\;\cap\;\complementof{Z_i}}}                     \\
			 & =
			\Prob[\game][\optzstrat\lrb{\initmem},\ostrat,\s]{\eventually\,
				\always\, U_{i+1}}.                                                \\
		\end{align*}
		Since $U_1 = \states$, $\eventually\,\always\, U_1$ is a sure event and hence
		\begin{align*}
			\Prob[\game][\optzstrat\lrb{\initmem},\ostrat,\s]{\bigcap_i\complementof{\obj_i}}
			 & =
			\Prob[\game][\optzstrat\lrb{\initmem},\ostrat,\s]{\eventually\,\always\,
				U_1\;\cap\; \complementof{\obj_1}\;\cap\;
				\bigcap_{i>1}\complementof{\obj_i}} \\
			 & =
			\Prob[\game][\optzstrat\lrb{\initmem},\ostrat,\s]{\eventually\,\always\,
				U_2\;\cap\; \complementof{\obj_2}\;\cap\;
				\bigcap_{i>2}\complementof{\obj_i}} \\
			 & =
			\Prob[\game][\optzstrat\lrb{\initmem},\ostrat,\s]{\eventually\,\always\,
				U_j\;\cap\; \complementof{\obj_j}\;\cap\;
				\bigcap_{i>j}\complementof{\obj_i}} \\
			 & =
			\Prob[\game][\optzstrat\lrb{\initmem},\ostrat,\s]{\eventually\,\always\,
				U_{l+1}}                            \\
			 & = 0
		\end{align*}
	\end{claimproof}

	\begin{remark}
		It is easy to see that the only property of $\MP > 0$ used
		in the proof was that it is shift-invariant. Hence,
		the above claim can be generalized to any objective of the form
		$\obj \, \cap \, \Eparity$ with $\obj$ being shift-invariant.
		However, the claim is presented in its present form since this
		observation by itself is of little significance if one cannot
		generalize the case with the highest color being even as well to utilize
		the induction argument.
	\end{remark}

    \maxcoloureven*
    
    \begin{claimproof}[Proof of~\cref{claim:max_colour_even}]
		For~\cref{itm:equivalence_even}, one direction (forward) is
		trivial. For
		the other direction (converse), by our assumptions, there exist
		\begin{enumerate}
			\item A uniform almost surely winning Max strategy $\zstrat_{\MP}$ that is MD such that
			      against any Min strategy $\ostrat$ and from any start state $\s$, $$
				      \Prob[\game][\zstrat_{\MP},\ostrat,\s]{\MP > 0} = 1 $$

			\item A almost surely winning Max strategy $\zstrat_{Y} \eqdef \memstrattuple$ such that
			      against any Min strategy $\ostrat \, \in \, \oallstrats{\game\lrb{Y}}$ and from
			      any start state $\s \, \in \, Y$, $$
				      \Prob[\game\lrb{Y}][\optzstrat_Y\lrb{\initmem},\ostrat,\s]{\mppar} = 1 $$

			\item A uniform positive attractor Max strategy $\zstrat_{\attr}$ that is MD defined for
			      all Max states in $X$ such that against any Min strategy $\ostrat$ and any
			      state $\s \, \in \, X$, $$
				      \Prob[\game][\zstrat_{\attr},\ostrat,\s]{\eventually\lrc{\states\lrc{d}}} > 0 $$

		\end{enumerate}
		For small probabilities $\eps_0,\eps_1 >0$,
                we define a parameterized family of finite-memory Max strategies $\optzstrat_{\eps_0,
			\eps_1} \eqdef \tuple{\memconf_{\states} \uplus \memconfset, \memup^{\prime},
			\memsuc^{\prime}}$ where

		\begin{align*}
			\memsuc^\prime\lrc{\memconf,\s}            & \eqdef \begin{cases}
				\eps_0 \cdot \zstrat_{\attr}\lrc{\s} + \lrc{1-\eps_0} \cdot \zstrat_{\MP}\lrc{\s} & \s\,\in\, X                                    \\
				\zstrat_{\MP}\lrc{\s}                                                             & \s\, \in \, Y,\; \memconf = \memconf_{\states} \\
				\memsuc\lrc{\memconf,\s}                                                          & \s\,\in\, Y,\; \memconf\, \in \, \memconfset
			\end{cases} \\
			\memup^\prime\lrc{\memconf,\tuple{\s,\s'}} & \eqdef \begin{cases}
				\memup\lrc{\memconf,\tuple{\s,\s'}}                               & \memconf \, \in \, \memconfset,\; \s\text{ and }\s' \, \in \, Y \\
				\memconf_{\states}                                                & \memconf \, \in \, \memconfset,\; \s\text{ or }\s'\, \in\, X    \\
				\eps_1 \cdot \initmem + \lrc{1-\eps_1} \cdot \memconf_{\states} & \memconf\,=\,\memconf_{\states},\; \s' \, \in \, Y              \\
				\memconf_{\states}                                                & \memconf\,=\,\memconf_{\states},\; \s' \, \in \, X
			\end{cases}
		\end{align*}

		If $\zstates \, \cap \, X$ is empty, $\eps_0$ is redundant. If $\zstates \,
		\cap \, Y$ is empty, $\memconfset$ and $\eps_1$ are redundant. 
		Observe that when there is only one even color, then indeed $Y$
		must be empty, because otherwise Max cannot win almost surely from $Y$.
		In this case, $\bits(\optzstrat_{\eps_0,\eps_1})$ is determined by how small
		$\eps_0$ has to be. We show below that exponentially small $\eps_0$
		suffices. For $\eps_1$, we assume that $\numevencolours$ is at least $2$. 
		We argue that there is an instantiation with sufficiently small (doubly exponentially small numbers
		suffice as we will see) positive values for $\eps_1$ such that
		$\optzstrat_{\eps_0, \eps_1}$ is almost surely winning from any
		start state. Specifically, we show that it is possible to instantiate 
		$\eps_0$ and $\eps_1$ such that
		\begin{align}
			\size{\eps_0} &= f_2\lrc{n,p_0} \label{eqn:size_eps0} \\
			\size{\eps_1} &= \Ocompl\lrc{\size{\game}^{c_0} + \bits(\optzstrat_Y)} \label{eqn:size_eps1}
		\end{align}
		where $f_2$ is some polynomial function in two variables, $n$ is the
		number of states in $\game$, $p_0$ is the bit size of the 
		probability transition function in $\game$ and $c_0$ is some constant independent
		of the instance.
 		From the definition of $\optzstrat$, one has $\size{\optzstrat} =
		\size{\optzstrat_Y} + 1$.
		\begin{align*}
			\bits(\memsuc^{\prime}) &= \lrc{\size{\eps_0} + \size{1-\eps_0}} 
			\card{X} (\size{\optzstrat_Y} + 1) + \bits(\memsuc) \\
			\bits(\memup^{\prime}) &= \bits(\memup) + \lrc{\size{\eps_1} + 
			\size{1-\eps_1}} \card{Y}
		\end{align*}
		
		Substituting for $\size{\eps_0}$ and $\size{\eps_1}$ from above, we get
		\begin{align*}
			\bits(\optzstrat) &= \Ocompl\lrc{f_2(n,p_0) \size{\optzstrat_Y} 
			 \card{\states} + \bits(\optzstrat_Y) + \lrc{\size{\game}^{c_0} + 
			 \bits(\optzstrat_Y)} \card{\states}} \\
			&= \Ocompl\lrc{\size{\game}^{c_1} \size{\optzstrat_Y} + 
			 \bits(\optzstrat_Y) \size{\game}}
		\end{align*}

		Both size and bit complexity of $\optzstrat$ satisfy the
		properties from~\cref{itm:strategy_even}. The proof of claim is 
		complete once we show the required bounds and that $\optzstrat$
		is almost surely winning $\mppar$ from every state.
		Observe that $$ \forall \ostrat\;
			\Prob[\game][\optzstrat\lrb{\initmem},\ostrat,\s]{\mppar} = 1 \iff \forall \ostrat\;
			\Prob[\game][\optzstrat\lrb{\initmem},\ostrat,\s]{\Eparity} = 1 \wedge
			\Prob[\game][\optzstrat\lrb{\initmem},\ostrat,\s]{\MP > 0} = 1.$$ 
		For $\Eparity$, it suffices to have $\eps_0 > 0$ since
		$\Prob[\game][\optzstrat,\ostrat,\s]{\Eparity} = \Prob[\game][\optzstrat\lrb{\initmem},\ostrat,\s]{\Eparity \, \cap \,  \always\,\eventually\,X} + \Prob[\game][\optzstrat\lrb{\initmem},\ostrat,\s]{\Eparity \, \cap \,  \complementof{\always\,\eventually\,X}}$
		\begin{align*}
			 & = \Prob[\game][\optzstrat\lrb{\initmem},\ostrat,\s]{\Eparity \, \cap \,  \always\,\eventually\,X}                     & + & \Prob[\game][\optzstrat\lrb{\initmem},\ostrat,\s]{\complementof{\always\,\eventually\,X}} & \optzstrat\text{ eventually behaves as }\optzstrat_Y                                      \\
			 & = \Prob[\game][\optzstrat\lrb{\initmem},\ostrat,\s]{\Eparity \, \cap \,  \always\,\eventually\,\states\lrc{d}} & + & \Prob[\game][\optzstrat\lrb{\initmem},\ostrat,\s]{\complementof{\always\,\eventually\,X}} & \Cref{claim:attr_properties}~\cref{itm:attractor_target_equiv} \text{ as } \eps_0 > 0 \\
			 & = \Prob[\game][\optzstrat\lrb{\initmem},\ostrat,\s]{\always\,\eventually\,\states\lrc{d}}                      & + & \Prob[\game][\optzstrat\lrb{\initmem},\ostrat,\s]{\complementof{\always\,\eventually\,X}}                      &                                                                                                   \\
			 & = \Prob[\game][\optzstrat\lrb{\initmem},\ostrat,\s]{\always\,\eventually\,X}                                          & + & \Prob[\game][\optzstrat\lrb{\initmem},\ostrat,\s]{\complementof{\always\,\eventually\,X}}                      & ~\Cref{claim:attr_properties}                                                                     \\
			 & = 1 &   &   &
		\end{align*}
		For $\MP > 0$, our argument uses the fact that the strategy
		$\optzstrat$ has finite memory and the size of the
                probabilities it uses.
                Firstly, observe that
		$\game\lrb{Y}$ is a mean-payoff-parity game with highest color $< d$. 
		\ignore{Hence, by induction hypothesis, it is possible to choose $\optzstrat_Y$ such
		that $\size{\memconfset} \le k - 1$ and the size of the probabilities to be
		$p_{Y} = \Ocompl\lrc{n^{k-1+c} \cdot \lrc{k-3}!} $.}
		Now, to show that $\Prob[\game][\optzstrat,\ostrat,\s]{\MP > 0} = 1$
		against any strategy $\ostrat$ of Min, it suffices to consider the induced
		minimizing MDP $\mdp \eqdef \game^{\optzstrat}$ and show that the
		value of every state $\tuple{\s,\initmem}$ in $\mdp$ is $1$.
        Since the original game has finitely many states and $\optzstrat$ has
        finite memory, $\mdp$ is a finite-state MDP. By standard results on
		finite-state MDPs for $\MP > 0$~\cite{Brazdil2010}, it suffices to consider
		just MD strategies for Min in this MDP, resulting in a finite-state
		Markov chain.



		Fix some MD strategy $\ostrat_{\mdp}$ of Min in $\mdp$ resulting in the Markov
		chain $\mc\lrb{\optzstrat, \ostrat_{\mdp}}$ (henceforth referred to as $\mc$).
		Given a state $\s \, \in \, \states$, let $\?{B}$ be some reachable BSCC of
		$\mc$ when starting from $\tuple{\s, \initmem}$. Observe that as long
		as $\eps_0$ and $\eps_1$ are in $(0,1)$, the BSSC
		structure itself is independent of the values of $\eps_0, \eps_1$ and hence
		$\?{B}$ can be chosen prior to fixing either value. From the definition of
		$\optzstrat$, the memory part of the state is $\memconf_{\states}$
		whenever the state belongs to $X$, except possibly at the beginning.
		This implies that in $\mc$, every state in $X \x \memconfset$ is a
		transient state or unreachable, hence cannot be part of $\?{B}$. $\?{B}$ can therefore
		be seen as disjoint union of $B_1 \eqdef \states_{\?{B}} \, \cap \, 
		 \states \x \memconf_{\states}$ and $B_2 \eqdef \states_{\?{B}} \, 
		 \cap \,Y \x \memconfset$. While $B_1$ and $B_2$ are disjoint, there
		will nevertheless be transitions from $B_1$ to $B_2$ and vice-versa,
		if both are non-empty as $\?{B}$ is strongly connected.

		We show that $\MP > 0$ by a case by case basis on whether either of $B_i$
		is empty. 
		If $B_2$ is empty, then $\?{B} = B_1$ and this implies that the
		`memory part' of the state is $\memconf_{\states}$.
		Here $\optzstrat$ behaves as $\zstrat_{\MP}$ with a perturbation
		towards $\zstrat_{\attr}$.
		But one cannot directly claim that this perturbation only causes a
		small change in the average mean-payoff of each state.
		This is because the perturbation itself might cause for some BSCC's
		to merge together or make them transient.
		These are known as singular perturbations~\cite{avrachenkov2013analytic,
		avrachenkov2002singular,lasserre1994formula} within the general area of
		research of perturbation theory~\cite{schweitzer1968perturbation,
		delebecque1983reduction}.
		When the stationary distribution is viewed as a function of the
		perturbation $\eps_0$, it is known that singular perturbations might
		exhibit a discontinuity at $\eps_0 = 0$~\cite[Section~6.1]{avrachenkov2013analytic}.
		Nevertheless, the stationary distribution of the BSCC's in perturbed
		chain are known to vary as $\pi(\eps_0) = \pi_0 + \pi_1 \eps_0 + \pi_2
		\eps_0^2 + \ldots$ with $\pi_0$ being a convex combination of distributions
		of the original BSCC's~\cite[Eq.~6.20]{avrachenkov2013analytic}.
		This implies that the mean-payoff is also close to some convex
		combination of mean-payoff's in the original BSCC's for small $\eps_0$.
		Thus, one can bound the size of $\eps_0$ to be some polynomial to
		ensure that $\MP > 0$.
		We formalize these arguments below.
		Let $\pi^{B_1}(\eps_0)$ denote the stationary distribution in $B_1$ as
		a function of $\eps_0$.
		Denote by $C_1 \ldots C_k \subseteq B_1$ the BSCC's in the original
		chain ($\eps_0 = 0$) with corresponding steady state distributions
		$\pi^{C_i}$. 
		Note that all $\pi$ are vectors but we omit the  $\vec{\cdot}$ notation
		for simplicity.
		\begin{claim}\label{claim:perturbed-stationary-bound}
			There are computable constants
			$\pi_0^{B_1}$, $\pi_1^{B_1} \in \Q^{B_1}$, $K^{B_1} \geq 0$,
			$\eps^{B_1} \in \Q$, all of which have sizes polynomial in $\size{B_1}$ such
			that $\forall\, \eps_0 \in [0,\eps^{B_1}]$
			\begin{equation} \label{eq:error-term-bound}
				\norm[\infty]{\pi^{B_1}(\eps_0) - \lrc{\pi_0^{B_1} +
				 \pi_1^{B_1} \eps_0 }} \leq K^{B_1}\, \eps_0^2
			\end{equation}
			Moreover,
			\begin{align}
				\pi_0^{B_1} &= \pi^{(1)}_0 C \label{eq:pi-0} \\
				\pi^{B_1}_1 &= \pi^{(1)}_1 C + \pi^{B_1}_0 \probp_1 H \label{eq:pi-1}
			\end{align}
			where $C(i,\s) \eqdef \mathbf{1}_{\s \in C_i} \cdot \pi^{C_i}(\s)$ is a
			$k \x B_1$ matrix, $\pi^{(1)}_i$ is another sequence which can be
			seen as coefficients of stationary distribution of a smaller chain,
			$H$ is the deviation matrix of $\probp^{B_1}_{\MP}$.
		\end{claim}
		\begin{claimproof}
			Let $\probp^{B_1}_{\MP}$ denote the transition matrix of
			$\mc\lrb{\zstrat_{\MP}, \ostrat_{\mdp}}$ restricted to $\?{B} = B_1$,
			then the probability transition matrix for $\mc$ restricted to $B_1$
			can be seen as $\probp^{B_1} = \probp^{B_1}_{\MP} + \eps_0 \probp_1$, where
			$\probp_1$ has non-zero rows only for states $\s \in \zstates \cap
			X$ such that $\zstrat_{\attr}(\s) \neq \zstrat_{\MP}(\s)$ and in
			this case $\probp_1\lrc{\tuple{\s,\zstrat_{\attr}(\s)}} = 1$,
			$\probp_1\lrc{\tuple{\s, \zstrat_{\MP}(\s)}} = -1$. 
			Also note that $B_1$ is a sub-Markov chain of $\probp_{\MP}$,
			since the perturbation can only add edges, and if there is no
			probability leaking to states outside $B_1$ from $B_1$ in the
			perturbed matrix ($B_1 = \?{B}$ is a BSCC), then this is also
			the case in the unperturbed matrix. But $\probp^{B_1}_{\MP}$ need
			not be irreducible, instead let $C_1 \ldots C_k \subseteq B_1$ be
			it's BSCC's. 
			This is the setting considered in~\cite[Section~6.2.1]{avrachenkov2013analytic}.
			From~\cite[Proposition~ 6.1]{avrachenkov2013analytic} and the
			preceding text, one can express $\pi^{B_1}(\eps_0)$ as an infinite
			series 
			\begin{equation} \label{eq:avrachenkov-expansion}
				\pi^{B_1}_0 + \pi^{B_1}_1 \eps_0 + \pi^{B_1}_2 \eps_0^2 + \ldots.
			\end{equation}
			The first two terms can also be explicitly given as~\eqref{eq:pi-0}, 
			\eqref{eq:pi-1}from the same proposition.
			Moreover, since all the entities $\pi^{(1)}_0, \pi^{(1)}_1, C, H$
			are either solutions of linear equations with entries whose size
			is bounded by some polynomial or computed using such matrices,
			$\size{\pi^{B_1}_0}, \size{\pi^{B_1}}$ is polynomial.
			To now bound the higher order $\eps_0$ terms, we instead go by the
			Taylor series expansion of $\pi^{B_1}(\eps_0)(\s)$ for each state
			$\s$.
			The stationary distribution satisfies the equations
			\[
				\pi^{B_1}(\eps_0) (\probp^{B_1} - I)  = \vec{0}, \pi^{B_1}(\eps_0) \vec{1} = 1
			\]
			Symbolically solving these linear system of equations (for instance
			by using Cramer's rule) gives us that each component $\pi^{B_1}
			(\eps_0)(\s)$ is some rational function $f_{\s}(\eps_0) = \frac{
			p_{\s}(\eps_0)}{q_{\s}(\eps_0)}$, where both $p_\s$ and $q_{\s}$ are
			polynomials of degree at most $\card{B_1}$ (assume both $p_{\s}$
			and $q_{\s}$ are in reduced form, i.e., any common factors are canceled).
			Furthermore, since the previous discussion guarantees that $\lim_{
			\eps_0 \tendsto 0}f_{\s}(\eps_0)$ exists, one can deduce that $q_{\s}$
			doesn't vanish at $0$ since $p_\s$ and $q_{\s}$ are in reduced form. 
			This implies that $f_{\s}$ is analytic in an open disk around $0$,
			and thus the Taylor series of $f_{\s}$ converges to $f_{\s}$ in this
			disk.
			Consequently, the Taylor's theorem guarantees that for every 
			$\eps_0 \in (0,\eps^{B_1})$,
			\begin{equation} \label{eq:taylor-expansion}
    			\pi^{B_1}(\eps_0)(\s) = f_{\s}(0) + f_{\s}'(0)\eps_0 + 
    			\frac{f_{\s}''(\xi)}{2}\eps_0^2 \quad \text{for some } \xi \in (0, \eps_0).
			\end{equation}
			We choose $\eps^{B_1}$ small enough so that $q_{\s}$ doesn't vanish in
			the above range and every point is sufficiently far away from any
			root of $q_{\s}$. 
			That the size of $\eps^{B_1}$ can be chosen to be polynomial is
			established by observing that the magnitude of any root of $q_{\s}$
			is lower bounded by some polynomial.
			Therefore, in this interval $\abs{f_{\s}''}$ is bounded by some
			constant $K_\s$ and size of $K_\s$ is polynomially bounded since
			$f_{\s}''$ is again a rational function with denominator bounded
			away from $0$.
			Let $K^{B_1} = \max_{\s \in B_1} K_\s$ whose size is polynomial.
			Since both~\eqref{eq:taylor-expansion} and~\eqref{eq:avrachenkov-expansion}
			initiated at $\s$ represent the same function, matching the terms
			gives us that $f_{\s}(0) = \pi_0^{B_1}$ and $f_{\s}'(0) = 
			\pi_1^{B_1}$.
			So, defining $f = (f_\s)_{\s \in B_1}$
			\begin{align*}
				\norm[\infty]{\pi^{B_1}(\eps_0) - \lrc{\pi_0^{B_1} +
				\pi_1^{B_1} \eps_0 }} &= \norm[\infty]{\pi^{B_1}(\eps_0)
				 - \lrc{f(0) +
				  \eps_0 f'(0)}} \\
				  & \leq \norm[\infty]{\frac{f''(\vec{\xi})}{2}\eps_0^2} \\
				  & \leq K^{B_1}/2 \eps_0^2
			\end{align*}
		\end{claimproof}

		Let $\mu^{B_1}(\eps_0)$ denote the mean-payoff in $B_1$ as a function
		of $\eps_0$.
		\begin{claim} \label{claim:size-of-eps-0}
			There exist $\eps$, with a polynomially bounded size, such that
			$\mu^{B_1}(\eps) > 0$.
		\end{claim}
		\begin{claimproof}
			Using~\cref{claim:perturbed-stationary-bound}, one can similarly
			construct an infinite series for $\mu^{B_1}(\eps_0)$. Since as $\eps_0
			\tendsto 0$, $\pi^{B_1}(\eps_0) = \pi^{B_1}_0 = \pi^{(1)}_0 C$,
			$\lim_{\eps_0 \tendsto 0} \mu^{B_1}(\eps_0) = (\pi^{(1)}_0) M$
			where $M = (\mu^{C_i})_{1 \leq i \leq k}$ is a $k \x 1$ vector
			capturing the mean-payoffs of the BSCC's $C_i$ in the chain
			$\mc\lrb{\zstrat_{\MP}, \ostrat_{\mdp}}$.
			By properties $\zstrat_{\MP}$, each $\mu^{C_i} \geq \mu_{\min}$
			for some $\mu_{\min} > 0$ of polynomial size and hence the limit is also
			$\geq \mu_{\min}$. 
			This combined with~\eqref{eq:error-term-bound} imply that the mean-payoff
			remains $> 0$ and of polynomial size for some $\eps_0 = \eps$ of
			polynomially bounded size.
		\end{claimproof}

		On the other hand, when $B_1$ is empty, $\optzstrat$ behaves as
		$\optzstrat_Y$ and $B_1$ being empty implies $\ostrat_{\mdp}$ is also a valid
		strategy in $\game\lrb{Y}$. By properties of $\optzstrat_Y$, this means $\MP > 0$
		almost surely.

		If both $B_1$ and $B_2$ are non-empty, then we argue by deriving lower bounds
		on the expected sum of rewards in $B_1$ and $B_2$ under steady state
		distribution. Let $X^{\s}_{i,\?B}$ denote the state at time $i$ and
		$Y^{\s}_{i,\?B}$ denote the random variable which computes the sum of the
		rewards until step $i$ when starting from $\s \, \in \, \?B$. Similarly let
		$T^{\s}_{2,\?B}$ (resp. $T^{\qstate}_{1,\?B}$) denote the hitting times of
		$B_2$ (resp. $B_1$) when starting from $\s \, \in \, B_1$ (resp. $\qstate \,
		\in \, B_2$). Note that although the states of $\?B$ are tuples, we
		refer to them by $\s$ and $\qstate$ for notational simplicity.
		\begin{align}
			Y^{\s}_{i,\?B} &\eqdef \sum_{j=0}^{i-1} r\lrc{\tuple{X^{\s}_{j,\?B}, X^{\s}_{j+1,\?B}}} & \text{for all } \s\,\in\,\?B \text{, } i \ge 0 \label{eqn:Y_rv_defn}
		\\  T^{\s}_{2,\?B} &\eqdef \min\setcomp{i}{X^{\s}_{i,\?B} \, \in \, B_2} & \s \, \in \, B_1 \label{eqn: T2_rv_defn}
		\\  T^{\qstate}_{1,\?B} &\eqdef \min\setcomp{i}{X^{\qstate}_{i,\?B} \, \in \, B_1} & \qstate \, \in \, B_2
		\end{align}
		If one can find uniform constants $v_1 > 0$ and $v_2$ such that
		\begin{align}
			\forall\, \s\quad       & \expectation[\?B][\s]{Y^{\s}_{T^{\s}_{2,\?B},\?B}}           & \ge v_1 \label{eqn:hyp_expectation_lower_bound_B1}
		\\	\forall\, \qstate \quad & \expectation[\?B][\qstate]{Y^{\qstate}_{T^{\qstate}_{1,\?B},\?B}} & \ge v_2 \label{eqn:hyp_expectation_lower_bound_B2}
		\\	                        & v_1 + v_2                              & > 0 \label{eqn:satifying_condition}
		\end{align}
		then it suffices since $\MP$ almost surely equals $\expectation{\MP}$ within a
		BSCC and
		\begin{align*}
			\expectation[\?{B}]{\MP} & = \frac{\sum_{\s} \lambda_{\s} Y^{\s}_{T^{\s}_{2,\?B},\?B} + \sum_{\qstate} \nu_{\qstate} Y^{\qstate}_{T^{\qstate}_{1,\?B},\?B}}{\sum_{\s} \lambda_{\s} T^{\s}_{2,\?B} + \sum_{\qstate} \nu_{\qstate} T^{\qstate}_{1,\?B}} \\
			                         & \ge \frac{\sum_{\s} \lambda_{\s} v_1 + \sum_{\qstate} \nu_{\qstate} v_2}{\sum_{\s} \lambda_{\s} T^{\s}_{2,\?B} + \sum_{\qstate} \nu_{\qstate} T^{\qstate}_{1,\?B}}        \\
			                         & \ge \frac{v_1 + v_2}{\sum_{\s} \lambda_{\s} T^{\s}_{2,\?B} + \sum_{\qstate} \nu_{\qstate} T^{\qstate}_{1,\?B}}                                                            \\
			                         & > 0
		\end{align*}
		where $\lambda_\s$ (resp. $\nu_{\qstate}$) are long term steady state
		probabilities of hitting $B_1$ (resp. $B_2$) at $\s$ (resp. $\qstate$).

		We have $\expectation[\?B][\qstate]{Y^{\qstate}_{T^{\qstate}_{1,\?B},\?B}}
		\ge -R \expectation[\?B][\qstate]{T^{\qstate}_{1,\?B}}$. To upper bound
		$T^{\qstate}_{1,\?B}$, observe that $\card{B_2} \le \card{Y} \x
		\card{\memconfset}$. Let $p_Y$ denote the maximum size of any 
		probability used by $\optzstrat_Y$ and $x_Y$ denote the smallest
		probability in $\mc$ restricted to $Y \x \memconfset$. It is easy
		to see that $x_Y \ge 2^{-p_Y}$ and 
		\begin{equation} \label{eqn:encoding_length_to_bit_size}
			\bits(\optzstrat_Y) = 
			\Ocompl(\card{Y} \size{\optzstrat_Y} p_Y). 
		\end{equation}
		The probability that a state with 
		transition to $B_1$ is hit in the first $\card{B_2}$ steps is at least
		${x_Y}^{\card{B_2}}$ from any start state, $\ie$ continuing for every $b$
		steps of that size we have
		\[
			\forall \qstate\, \in \, B_2 \; \Prob[\?B][\qstate]{T^{\qstate}_{1,\?B} \ge  b \cdot
				\card{B_2}} \le \lrc{1-{x_Y}^{\card{B_2}}}^b
		\]
		This gives $\expectation[\?B][\qstate]{T^{\qstate}_{1,\?B}} \le
		{x_{Y}}^{-\card{B_2}} \cdot \card{B_2}$ Therefore
		\begin{equation} \label{eqn:expectation_sum_in_B2}
			\forall \qstate\, \in \, B_2 \; \expectation[\?B][\qstate]{Y^{\qstate}_{T^{\qstate}_{1,\?B},\?B}} \ge - {x_{Y}}^{-\card{B_2}} \cdot \card{B_2} \cdot R
		\end{equation}

		To find a lower bound for $\expectation[\?B][\s]{Y^{\s}_{T^{\s}_{2,\?B},\?B}}$,
		we need to link the reward gained in $B_1$ in the Markov chain $\mc$ to
		reward gained in $B_1$ in the Markov chain $\mc^{\prime} \eqdef
		\game\lrb{\zstrat_1,\ostrat^{\prime}}$ where $\zstrat_1$ and $\ostrat^{\prime}$
		are memoryless restrictions of $\optzstrat$ and $\ostrat_{\mdp}$ to $\states \x
		\memconf_{\states}$. Also note that $\mc^{\prime}\lrb{B_1}$ 
		is a well-defined sub-Markov chain and is equivalent to the conditioned
		Markov chain of $\mc$ on $B_1$.

		\begin{definition}[Conditioned Markov chain]\label{defn: conditioned_mc}
			If $\mc \eqdef \tuple{\states, \transition, \probp}$ is some 
			finite-state Markov chain and $\states^{\prime} \subseteq \states$ such 
			that for all $\s^{\prime} \, \in \, \states^{\prime}$, 
			$\sum_{\s \in \states^{\prime}} \probp\lrc{\tuple{\s^{\prime},\s}} 
			 > 0$, then the conditioned Markov chain is well defined and given 
			by $\mc_{\states^{\prime}} \eqdef \tuple{\states^{\prime}, \transition 
			 \, \cap \, \lrc{\states^{\prime} \x \states^{\prime}}, \probp_{\states^{\prime}}}$, 
			where $\probp_{\states^{\prime}}\lrc{\tuple{\s^{\prime}, \s}} \eqdef 
			 \frac{\probp\lrc{\tuple{\s^{\prime},\s}}}{\sum_{\s \in \states^{\prime}} 
			 \probp\lrc{\tuple{\s^{\prime},\s}}}$
		\end{definition}

		Denote by $r^{\?B}_{\ell} \eqdef r\lrc{\tuple{X^{\s}_{l,\?B},
				X^{\s}_{l+1,\?B}}}$ the random variable which gives the reward on the
		$\ell^{th}$ step when starting from $\s \in \?B$. For an event $A$, let
		$\mathbf{1}_{A}$ denote the indicator random variable of $A$. Simplifying, we
		get $\expectation[\?B][\s]{Y^{\s}_{T^{\s}_{2,\?B},\?B}}$
		\begin{align*}
			&= \expectation[\?B][\s]{ \sum_{\ell=0}^{\infty} r^{\?B}_{\ell} \cdot \mathbf{1}_{ T^{\s}_{2,\?B} > \ell } } & \text{from}~\eqref{eqn:Y_rv_defn}
			\\ &= \sum_{\ell=0}^{\infty} \expectation[\?B][\s]{r^{\?B}_{\ell} \cdot \mathbf{1}_{ T^{\s}_{2,\?B} > \ell }} & \abs{r^{\?B}_{\ell}} \le R \text{ , } \expectation[\?B][\s]{T^{\s}_{2,\?B}} < \infty
			\\ &= \sum_{\ell=0}^{\infty} \expectation[\?B][\s]{r^{\?B}_{\ell} \cdot \mathbf{1}_{ T^{\s}_{2,\?B} > \ell + 1}} + \expectation[\?B][\s]{r^{\?B}_{\ell} \cdot \mathbf{1}_{ T^{\s}_{2,\?B} = \ell + 1 }} & T^{\s}_{2,\?B} \text{ is integer valued}
		\end{align*}

		The value $\expectation[\?B][\s]{r^{\?B}_{\ell} \cdot
		\mathbf{1}_{T^{\s}_{2,\?B} = \ell+1}}$ can be lower bounded by $- R \cdot
		\Prob[\?B][\s]{T^{\s}_{2,\?B} = \ell+1}$ since $r^{\?B}_{\ell} \ge -R$
		surely. Furthermore, it is easy to see from~\eqref{eqn: T2_rv_defn} that the
		event $T^{\s}_{2,\?B} > \ell + 1$ is exactly $\bigcap_{j=0}^{\ell+1}
		X^{\s}_{j,\?B} \, \in \, B_1$. Denoting the latter event by
		$\always^{[0,\ell+1]} B_1$ and continuing,
		\begin{align*}
			&\ge \sum_{\ell=0}^{\infty} \lrc{ \expectation[\?B][\s]{r^{\?B}_{\ell} \cdot \mathbf{1}_{\always^{[0,\ell+1]} B_1}} - R \cdot \Prob[\?B][\s]{T^{\s}_{2,\?B} = \ell+1} } & 
			\\ &\ge \lrc{\sum_{\ell=0}^{\infty} \expectation[\?B][\s]{r^{\?B}_{\ell} \cdot \mathbf{1}_{\always^{[0,\ell+1]} B_1}} } - R & \s \in B_1
		\end{align*}
		Intuitively, the expectation of the sum as long as one stays in $B_1$ should be
		the same as the expectation of the sum in the conditioned Markov chain
		$\mc_{B_1}$ (this is also the reason to throw away the last transition). We
		formalize this notion.
		\begin{claim}\label{claim:link_conditional_expectation_to_expectation_in_conditioned_mc}
			For all $\ell \ge 0$, $\s \, \in \, B_1$
			\[
				\expectation[\?B][\s]{r^{\?B}_{\ell} \cdot \mathbf{1}_{\always^{[0,\ell+1]} B_1}} = \expectation[\mc_{B_1}][\s]{r^{\mc_{B_1}}_{\ell}} \cdot \Prob[\?B][\s]{\always^{[0,\ell+1]} B_1}
			\]
		\end{claim}
		\begin{claimproof}[Proof of~\cref{claim:link_conditional_expectation_to_expectation_in_conditioned_mc}]
			For succinctness let $E_{\ell+1} = \always^{[0,\ell+1]} B_1$. The claim follows if one proves that for any $\s_1, \s_2 \, \in \, B_1$
			\[
				\Prob[\?B][\s]{X^{\s}_{\ell,\?B} = \s_1, X^{\s}_{\ell+1,\?B} = \s_2 \given E_{\ell+1} } = \Prob[\mc_{B_1}][\s]{X^{\s}_{\ell,\mc_{B_1}} = \s_1, X^{\s}_{\ell+1,\mc_{B_1}} = \s_2}
			\]
			holds. It is easy to see why as the rewards on the identical transitions are
			equal in both Markov chains. To show above, we first start by showing that the
			distribution of states at the $\ell^{th}$ step is identical in both scenarios.
			Given $\ell$ and start state $\s$, for all $\s_1 \, \in \, B_1$ let
			$\theta_{\s,\ell}\lrc{\s_1} \eqdef \Prob[\?B][\s]{X^{\s}_{\ell,\?B} = \s_1
				\given E_{\ell+1}}$ and $\theta^{\prime}_{\s,\ell}\lrc{\s_1} \eqdef
			\Prob[\mc_{B_1}][\s]{X^{\s}_{\ell,\mc_{B_1}} = \s_1}$.

			\begin{claim}\label{claim:link_conditional_distribution_to_distribution_in_conditioned_mc}
				For all $\ell \ge 0$, $\s \, \in \, B_1$
				\[
					\theta_{\s,\ell} = \theta^{\prime}_{\s,\ell}
				\]
			\end{claim}
			\begin{claimproof}[Proof of~\cref{claim:link_conditional_distribution_to_distribution_in_conditioned_mc}]
				We prove by induction on $\ell$.

				\textbf{Base case :} When $\ell = 0$, the event $X^{\s}_{1,\?B} \, \in \, B_1$ has positive probability by properties of $B_1$ and therefore $\theta_{\s,\ell}$ is well-defined for every start state $\s \, \in \, B_1$ and is equal to $\delta_{\s}$. This proves the base case.

				\textbf{Induction step :} Assume, $\theta_{\s,\ell} = \theta^{\prime}_{\s,\ell}$ for some $\ell \ge 0$. To show that it holds for $\ell+1$, $\theta_{\s,\ell+1}\lrc{\s_2}$
				\begin{align*}
					   &= \Prob[\?B][\s]{X^{\s}_{\ell+1,\?B} = \s_2 \given E_{\ell+1}} &
					\\ &= \sum_{\s_1 \in B_1} \Prob[\?B][\s]{X^{\s}_{\ell+1,\?B} = \s_2 \given E_{\ell+1} , X^{\s}_{\ell,\?B} = \s_1} \cdot \theta_{\s,\ell}\lrc{\s_1} &
					\\ &= \sum_{\s_1 \in B_1} \Prob[\?B][\s_1]{X^{\s_1}_{1,\?B} = \s_2 \given X^{\s_1}_{1,\?B} \, \in \, B_1} \cdot \theta^{\prime}_{\s,\ell}\lrc{\s_1} & \text{Markov, IH}
					\\ &= \sum_{\s_1 \in B_1} \frac{\probp_{\?B}\lrc{\tuple{\s_1,\s_2}}}{\sum_{\s_3\, \in \, B_1}\probp_{\?B}\lrc{\tuple{\s_1,\s_3}}} \cdot \theta^{\prime}_{\s,\ell}\lrc{\s_1} &
					\\ &= \sum_{\s_1 \in B_1} \probp_{B_1}\lrc{\tuple{\s_1,\s_2}} \cdot \theta^{\prime}_{\s,\ell}\lrc{\s_1} & \cref{defn: conditioned_mc}
					\\ &= \theta^{\prime}_{\s,\ell+1}\lrc{\s_2}
				\end{align*}
			\end{claimproof}

			Now, $\Prob[\?B][\s]{X^{\s}_{\ell,\?B} = \s_1, X^{\s}_{\ell+1,\?B} = \s_2
				\given E_{\ell+1} }$
			\begin{align*}
				&= \Prob[\?B][\s]{X^{\s}_{\ell+1,\?B} = \s_2 \given E_{\ell+1} , X^{\s}_{\ell,\?B} = \s_1} \cdot \theta_{\s,\ell}\lrc{\s_1} & 
			\\  &= \probp_{B_1}\lrc{\tuple{\s_1,\s_2}} \cdot \theta^{\prime}_{\s,\ell}\lrc{\s_1} & \cref{claim:link_conditional_distribution_to_distribution_in_conditioned_mc}
			\\  &= \Prob[\mc_{B_1}][\s]{X^{\s}_{\ell,\mc_{B_1}} = \s_1, X^{\s}_{\ell+1,\mc_{B_1}} = \s_2} &
			\end{align*}
			Therefore, one has $\expectation[\?B][\s]{r^{\?B}_{\ell} \given E_{\ell+1}} =
			\expectation[\mc_{B_1}][\s]{r^{\mc_{B_1}}_{\ell}}$ from which the claim
			follows.
		\end{claimproof}

		From~\cref{claim:link_conditional_expectation_to_expectation_in_conditioned_mc},
		we further get $\expectation[\?B][\s]{Y^{\s}_{T^{\s}_{2,\?B},\?B}}$
		\begin{align*}
			&\ge \lrc{\sum_{\ell=0}^{\infty} \expectation[\mc_{B_1}][\s]{r^{\mc_{B_1}}_{\ell} } \cdot \Prob[\?B][\s]{T^{\s}_{2,\?B} > \ell + 1}} - R &
		\\  &\ge \lrc{\sum_{\ell=0}^{\infty} \expectation[\mc_{B_1}][\s]{r^{\mc_{B_1}}_{\ell} } \cdot \sum_{t = \ell + 2}^{\infty} \Prob[\?B][\s]{T^{\s}_{2,\?B} = t } } - R & 
		\\  &\ge \lrc{\sum_{t=2}^{\infty} \Prob[\?B][\s]{T^{\s}_{2,\?B} = t } \cdot \lrc{ \sum_{\ell=0}^{t-2} \expectation[\mc_{B_1}][\s]{r^{\mc_{B_1}}_{\ell} } } } - R & \text{Interchange sums}
		\\ &\ge \lrc{\sum_{t=2}^{\infty} \Prob[\?B][\s]{T^{\s}_{2,\?B} = t } \cdot \expectation[\mc_{B_1}][\s]{Y^{\s}_{t-1,\mc_{B_1}}} } - R & \text{from}~\eqref{eqn:Y_rv_defn} \text{ and linearity}
		\\ &\ge \lrc{\sum_{t=1}^{\infty} \Prob[\?B][\s]{T^{\s}_{2,\?B} = t } \cdot \expectation[\mc_{B_1}][\s]{Y^{\s}_{t-1,\mc_{B_1}}} } - R & Y^{\s}_{0,\mc_{B_1}} = 0
		\end{align*}

		This allows us to turn a lower bound on sums in $\mc_{B_1}$ into
		a bound for $\expectation[\?B][\s]{Y^{\s}_{T^{\s}_{2,\?B},\?B}}$. Remember that
		$\mc_{B_1}$ is a sub-Markov chain of the Markov chain
		$\mc^{\prime} = \game\lrb{\zstrat_1,\ostrat^{\prime}}$ induced by two
		memoryless strategies which itself can be seen as perturbation of
		$\mc\lrb{\optzstrat, \ostrat_{\mdp}}$.
		This means, to guarantee that the mean-payoff is $> 0$, the size of the
		probabilities in $\mc^{\prime}$ denoted by $p_1 \eqdef
		 \max\lrc{\size{\eps_0}, p_0} = \size{\eps_0}$ can be bounded by some
		 polynomial using~\cref{claim:size-of-eps-0} and let the smallest
		probability be $x_1 \eqdef \eps_0$.
		A lower bound on sums in $\mc^{\prime}$ is also a lower bound for sums
		in $\mc_{B_1}$.
	{
		Let $\mc^{\prime} = \tuple{\states, \transition^{\prime}, \probp^{\prime}}$. 
		Fix a start state $\s$ and denote by $\mu^{\prime}$, the minimum achievable mean payoff in any BSCC in
		$\mc^{\prime}$. By optimality of $\zstrat_1$, $\mu^{\prime} > 0$. We can split
		the sum into two parts: the sum within a BSCC and the sum in transient states. For the
		former, within any BSCC $B$ of $\mc^{\prime}$, one can get a lower bound by
		standard arguments using the martingale process obtained from the Poisson
		equation.

		\begin{definition}
			For a function $f \, \in \, \R^{\states}$ on an ergodic Markov chain $\mc = \tuple{\states, \transition, \probp}$, the Poisson equation is given by
			\begin{equation*}
				\vec{f} + \probp \vec{h} = \vec{h} + \vec{1} \bar{f}
			\end{equation*}
			where $\bar{f} = \lim_{n \tendsto \infty} \frac{\sum_{i=0}^{n-1}f\lrc{X_{i,\mc}}}{n}$ is the long
			term average of $f$.
		\end{definition}

		By standard results~\cite[Theorem 3.4]{BKK2014}, a solution for the above
		equation exists and $\vec{h}$ can be chosen such that $h\lrc{\s}\, \in \, [0,
				K]$ where $K = \frac{2 \card{\states} f_{\max}}{x_0^{\card{\states}}}$,
		$f_{\max} = \max_{\s} \abs{f\lrc{\s}}$, $x_0$ the minimum non-zero probability
		in $\probp$. Moreover, given such a $h$, the process $M_n =
		\sum_{i=0}^{n-1}f\lrc{X_{i,\mc}} + h\lrc{X_{n,\mc}} - n \bar{f}$ is a
		martingale.

		To use these results in our context, lets fix a BSCC $B$ and a start state
		$\s_B$. Within this BSCC, let the long term average mean payoff be $\mu
		\ge \mu^{\prime} > 0$. Denote by $h^{B}_{\max} \eqdef \frac{2
			\card{\states_B} R^{B}}{{p^{B}_0}^{\card{\states_B}}} $. Then one can find
		$h^{B}: \states_{B} \to \lrb{0, h^{B}_{\max}}$ such that $M^{B}_t =
		Y^{\s_B}_{t,B} + h^{B}\lrc{X^{\s_B}_{t,B}} - t \mu$ is a martingale.

		$$\implies \expectation[B][\s_B]{Y^{\s_B}_{t,B} + h^{B}\lrc{X^{\s_B}_{t,B}} - t \mu} = \expectation[B][\s_B]{M^{B}_t} = \expectation[B][\s_B]{M^{B}_0} = h^{B}\lrc{\s_B} $$

		Simplifying, we get
		\[
			\forall\, \s_B\text{,}\; \expectation[B][\s_B]{Y^{\s_B}_{t,B}} = t \mu + h^{B}\lrc{\s_B} - \expectation[B][\s_B]{h^{B}\lrc{X^{\s_B}_{t,B}}} \ge t \mu^{\prime} - h^{B}_{\max}	
		\]

		Denote by $H$ the set of all states which are part of some BSCC in
		$\mc^{\prime}$ and $T^{\s}_{H} \eqdef \min\setcomp{i \ge
			0}{X^{\s}_{i,\mc^{\prime}} \, \in \, H}$ denote the hitting time to one of
		these BSCC's. $\card{\states} = n$, let $h^{\mc^{\prime}}_{\max} \eqdef
		 \max_{B\text{ is a BSCC }} h^{B}_{\max} \le \frac{2 n R}{x_1^n}$, then it
		is clear that for any $\s \, \in \, H$
		\begin{equation} \label{eqn:lower_bound_sum_bscc}
			\expectation[\mc^{\prime}][\s]{Y^{\s}_{t,\mc^{\prime}}} \ge 
			t \mu^{\prime} - h^{\mc^\prime}_{\max}
		\end{equation}
		and define $C_{\mc^{\prime}} \eqdef \lrc{R + \mu^{\prime}} \frac{n}{\lrc{x_1}^{n}} +
		\frac{2nR}{\lrc{x_1}^n} = \frac{3nR+n\mu^{\prime}}{\lrc{x_1}^n} $.
		~\eqref{eqn:lower_bound_sum_bscc} provides a lower bound for all the
		states in $H$. To get a lower bound on the sum for the transient states,
		we first compute an upper bound on the expected time spent in these
		states.

		\begin{claim}\label{claim:upper_bound_expected_time_transient}
			For any state $\s$ in $\mc^{\prime}$,
			$$ \expectation[\mc^{\prime}][\s]{T^{\s}_H} \le \frac{n}{x^{n}_1} $$
		\end{claim}

		\begin{claimproof}
			Much like in the derivation for~\eqref{eqn:expectation_sum_in_B2},
			one can show that $\Prob[\mc^{\prime}][\s]{T^{\s}_H > k \cdot n}
			  \le \lrc{1-x^n_1}^k$. This then implies
			\begin{align*}
    			\expectation[\mc^{\prime}][\s]{T^{\s}_H} &= \sum_{k=0}^{\infty} \Prob[\mc^{\prime}][\s]{T^{\s}_H > k}  \le n \sum_{k=0}^{\infty} \Prob[\mc^{\prime}][\s]{T^{\s}_H > k \cdot n} \\
    			&= n \sum_{k=0}^{\infty} \lrc{1-x^n_1}^k = \frac{n}{x^n_1}
			\end{align*}
		\end{claimproof}

		\begin{claim}\label{claim:lower_bound_sum_general_mc}
			$$\expectation[\mc^{\prime}][\s]{Y^{\s}_{t,\mc^{\prime}}} \ge t\mu^{\prime} - C_{\mc^{\prime}}$$
		\end{claim}

		\begin{claimproof}[Proof of~\cref{claim:lower_bound_sum_general_mc}]
			 We decompose the expected reward $\expectation[\mc^{\prime}][\s]{Y^{\s}_{t,\mc^{\prime}}}$ by conditioning on $T^{\s}_H$. For a given $T^{\s}_H=k$, the reward is a sum of rewards from the transient phase (at least $-R \cdot k$) and the recurrent phase. For the recurrent phase, the expected reward in the recurrent phase of length $t-k$ is at least $\lrc{t-k}\mu^{\prime} - h^{\mc^{\prime}}_{\max}$~(\cref{eqn:lower_bound_sum_bscc}).
			\begin{align*}
			    \expectation[\mc^{\prime}][\s]{Y^{\s}_{t,\mc^{\prime}}} &= \expectation[\mc^{\prime}][\s]{\expectation[\mc^{\prime}][\s]{Y^{\s}_{t,\mc^{\prime}} \given T^{\s}_H}} \\
			    &\ge \expectation[\mc^{\prime}][\s] {-R \cdot T^{\s}_H + (t-T^{\s}_H)\mu^{\prime} - h^{\mc^{\prime}}_{\max} } \\
			    &= t\mu^{\prime} - \lrc{R + \mu^{\prime}} \expectation[\mc^{\prime}][\s]{T^{\s}_H}  - h^{\mc^{\prime}}_{\max} \\
				&\ge t\mu^{\prime} - \lrc{R + \mu^{\prime}} \frac{n}{x^{n}_1} - h^{\mc^{\prime}}_{\max}~\lrc{\cref{claim:upper_bound_expected_time_transient}} \\
				&\ge t\mu^{\prime} - C_{\mc^{\prime}}
			\end{align*}
		\end{claimproof}

		Using~\cref{claim:lower_bound_sum_general_mc},
		$\expectation[\?B][\s]{Y^{\s}_{T^{\s}_{2,\?B},\?B}}$
		\begin{align*}
			&\ge \lrc{\sum_{t=1}^{\infty} \Prob[\?B][\s]{T^{\s}_{2,\?B} = t } \cdot  \lrc{t-1}\mu^{\prime} - C_{\mc^{\prime}} } - R &
		\\	&\ge \mu^{\prime} \cdot \lrc{\sum_{t=1}^{\infty} \Prob[\?B][\s]{T^{\s}_{2,\?	B} = t } \cdot  t} - C_{\mc^{\prime}} -\mu^{\prime}  - R & \s\,\in\,B_1
		\\  &= \mu^{\prime} \cdot \expectation[\?B][\s]{T^{\s}_{2,\?B}} - C_{\mc^{\prime}} -\mu^{\prime}  - R &
		\end{align*}

		It is easy to see that $\expectation[\?B][\s]{T^{\s}_{2,\?B}} \ge
		\frac{1}{\eps_1}$. Substituting this, we get
		\begin{equation} \label{eqn:expectation_sum_in_B1}
			\expectation[\?B][\s]{Y^{\s}_{T^{\s}_{2,\?B},\?B}} \ge \frac{\mu^{\prime}}{\eps_1} - C_{\mc^{\prime}} - \mu^{\prime} - R
		\end{equation}

		Comparing~\eqref{eqn:expectation_sum_in_B1}
		with~\eqref{eqn:hyp_expectation_lower_bound_B1},
		and~\eqref{eqn:expectation_sum_in_B2}
		with~\eqref{eqn:hyp_expectation_lower_bound_B2}, and looking
		at~\eqref{eqn:satifying_condition}, we require $\eps_1$ to be such that

		$$\frac{\mu^{\prime}}{\eps_1} - C_{\mc^{\prime}} - \mu^{\prime} - R -
			{x_{Y}}^{-\card{B_2}} \cdot \card{B_2} \cdot R > 0$$ 
		One can consider $\eps_1$ such that
		\begin{equation} \label{eqn:eps1-value}
		 \eps_1 = \frac{\mu^{\prime}}{\ceil{C_{\mc^{\prime}} + \mu^{\prime} + R +
			{x_{Y}}^{-\card{B_2}} \cdot \card{B_2} \cdot R} + 1}
		\end{equation}
		$\mu^{\prime}$ and $C_{\mc^{\prime}}$ constants which arise out of
		$\mc^{\prime}$ whose size is polynomial in $\game$. This implies
		that the sizes of both $\mu^{\prime}$ and $C_{\mc^{\prime}}$ is
		bounded by some polynomial with large enough degree. Also, $\card{B_2}
		 \le \card{Y} \x \size{\optzstrat_Y}$. Combining all the above facts, it is
		easy to see that $\size{\eps_1} = \Ocompl\lrc{\size{\game}^c 
		 + \card{Y} \cdot \size{\optzstrat_Y} \cdot p_{Y}}$ for some large enough
		degree $c$. From~\eqref{eqn:encoding_length_to_bit_size}, this shows
		that $\size{\eps_1} = \Ocompl\lrc{\size{\game}^{c_0} + 
		\bits(\optzstrat_Y)}$ which is what we sought to prove
		in~\eqref{eqn:size_eps1}.
		\ignore{
		\begin{align*}
			&> \lrc{3nR+n\mu^{\prime}} \cdot \alpha_1^n + \mu^{\prime} + R + {\alpha_{k-1}}^{n \cdot \lrc{k-1}} \cdot \lrc{n \cdot \lrc{k-1}} \cdot R \\
		\end{align*}
		}
	}	

	\end{claimproof}



\end{document}